\documentclass{article}
\usepackage[left=1in,top=1in,right=1in,bottom=1in,head=1in]{geometry}
\usepackage{amsfonts,amsmath,amssymb,amsthm}
\usepackage{verbatim,float,url,enumerate}
\usepackage{graphicx,subfigure,epsfig,psfrag}
\usepackage{bm,dsfont,color,appendix}
\usepackage{natbib}
\usepackage{latexsym,graphicx}
\newtheorem{theorem}{Theorem}
\newtheorem{lemma}{Lemma}
\newtheorem{remark}{Remark}
\newtheorem{corollary}{Corollary}

\newcommand{\argmax}{\mathop{\mathrm{argmax}}}

\def\P{\mathbb{P}}

\def\col{\mathrm{col}}

\def\sign{{\mathrm{sign}}}
\def\supp{\mathrm{supp}}

\def\R{\mathbb{R}}

\def\cP{\mathcal{P}}

\def\Exp{\mathrm{Exp}}

\def\RSS{\mathrm{RSS}}

\def\V{\mathcal{V}}
\def\weight{\omega}
\def\cp{\overset{P}{\rightarrow}}

\title{Exact Post-Selection Inference for Sequential Regression
  Procedures}
\author{Ryan J. Tibshirani$^1$  
\and Jonathan Taylor$^2$  
\and Richard Lockhart$^3$ 
\and Robert Tibshirani$^2$}
\date{\normalsize $^1$Carnegie Mellon University,
$^2$Stanford University,
$^3$Simon Fraser University}

\begin{document}
\maketitle

\begin{abstract}
We propose new inference tools for forward stepwise 
regression, least angle regression, and the lasso.  Assuming a
Gaussian model for the observation vector $y$, we first describe  
a general scheme to perform valid inference after any selection event 
that can be characterized as $y$ falling into
a polyhedral set.  This framework
allows us to derive conditional (post-selection) hypothesis tests
at any step of forward stepwise or least angle regression, or any step 
along the lasso regularization path, because, as it turns out,
selection events for these procedures can be expressed as polyhedral
constraints on $y$.  The p-values associated with these tests are
exactly uniform under the null distribution, in finite samples,
yielding exact type I error control. The tests can also be inverted to
produce confidence intervals for appropriate underlying regression
parameters.  The R package {\tt selectiveInference}, freely available
on the CRAN repository, implements the new inference tools
described in this paper. 
\end{abstract}

\section{Introduction}
\label{sec:intro}
We consider observations $y \in \R^n$ drawn from a Gaussian model   
\begin{equation}
\label{eq:mod}
y = \theta + \epsilon, \;\;\; \epsilon \sim N(0,\sigma^2 I),
\end{equation}
Given a fixed matrix $X \in \R^{n\times p}$ of predictor variables,
our focus is to provide inferential tools for methods   
that perform variable selection and estimation in an adaptive linear
regression of $y$ on $X$.  Unlike much of the related literature on
adaptive linear modeling, we 
{\it do not assume that the true model is itself linear}, i.e.,
we do not assume that $\theta=X\beta^*$ for a vector of true
coefficients $\beta^* \in \R^p$.  The particular
regression models that we consider in this paper are built from
sequential procedures that add (or delete) one variable at
a time, such as forward stepwise regression (FS), least angle
regression (LAR), and the lasso regularization path.  However, we
stress that the underpinnings of our approach extends well beyond 
these cases.  

To motivate the basic problem and illustrate our proposed 
solutions, we examine a data set of 67 observations and 8 
variables, where the outcome is the log PSA level of men who had
surgery for prostate cancer. The same data set was used to motivate
the covariance test in \citet{LTTT2013}.\footnote{The results for the
  naive FS test and the covariance test differ slightly from those
  that appear in \citet{LTTT2013}. We use a version of FS
  that selects variables to maximize the drop in residual sum of
  squares at each step; \citet{LTTT2013} use a version based on
  the maximal absolute correlation of a variable with the
  residual. Also, our naive FS p-values 
  are one-sided, to match the one-sided nature of
  the other p-values in the table, whereas \citet{LTTT2013} use
  two-sided naive FS p-values.  Lastly, we use an
  $\Exp(1)$ limit for the covariance test, and \citet{LTTT2013} use an 
  F-distribution to account for the unknown variance.}  
The first two numeric columns of Table \ref{tab:prostate} show the
p-values for regression coefficients of variables that enter the
model, across steps of FS. The first column shows the results of 
applying naive, ordinary t-tests to compute the significance of these
regression coefficients.  We see that the first four variables are
apparently significant at the 0.05 level, but this is suspect, as the
p-values do not account for the greedy selection of variables that is
inherent to FS.  The second column shows our new
selection-adjusted  p-values for FS, from a {\it truncated 
Gaussian (TG) test} developed in Sections \ref{sec:poly} and 
\ref{sec:seq_tests}.  These do properly account for the greediness:
they are conditional on the active set at each step, and now just
two variables are significant at the 0.05 level.   

\begin{table}[htb]
\centering
\begin{tabular}{|r|rr||r|rrr|}
  \hline
 & FS, naive & FS, TG & & LAR, cov & LAR, spacing & LAR, TG \\ 
\hline
lcavol & 0.000 & 0.000 & lcavol & 0.000 & 0.000 & 0.000 \\ 
  lweight & 0.000 & 0.027 & lweight & 0.047 & 0.052 & 0.052 \\ 
  svi & 0.019 & 0.184 & svi & 0.170 & 0.137 & 0.058 \\ 
  lbph & 0.021 & 0.172 & lbph & 0.930 & 0.918 & 0.918 \\ 
  pgg45 & 0.113 & 0.453 & pgg45 & 0.352 & 0.016 & 0.023 \\ 
  lcp & 0.041 & 0.703 & age & 0.653 & 0.586 & 0.365 \\ 
  age & 0.070 & 0.144 & lcp & 0.046 & 0.060 & 0.800 \\ 
  gleason & 0.442 & 0.800 & gleason & 0.979 & 0.858 & 0.933 \\ 
   \hline
\end{tabular}
\caption{\it Prostate cancer data example: p-values across steps of
  the forward stepwise (FS) path, computing using naive $t$-tests
  that do not account for greedy selection, and our new truncated
  Gaussian (TG) test for FS; also shown are p-values for the least
  angle regression (LAR) path, computed using the covariance test of 
  \cite{LTTT2013}, and our new spacing and TG tests for LAR.}      
\label{tab:prostate}
\end{table}

The last three numeric columns of Table \ref{tab:prostate} show
analogous results for the LAR algorithm applied to the prostate cancer
data (the LAR and lasso paths are identical here, as there were no
variable deletions).  The covariance test \citep{LTTT2013}, reviewed
in the Section \ref{sec:covtest}, measures the improvement in the LAR
fit due to adding a predictor at each step, and the third column shows   
p-values from its $\Exp(1)$ asymptotic null distribution.  Our 
new framework applied to LAR, described in Section
\ref{sec:seq_tests}, produces the results in the rightmost 
column.  We note that this TG test assumes far 
less than the covariance test.  In fact, our TG
p-values for both FS and LAR do not require assumptions about the
predictors $X$, or about the true model being linear.  They also use a
null distribution that is correct in finite samples, rather than
asymptotically, under Gaussian errors in \eqref{eq:mod}.   The fourth
column above shows a computationally efficient approximation to
the TG test for LAR, that we call the {\it spacing test}.  Later, we 
establish an asymptotic equivalence between our new spacing for LAR
and the covariance test, and this is supported by the similarity
between their p-values in the table.    

The R package {\tt selectiveInference} provides an implementation of
the TG tests for FS and LAR, and all other inference tools
described in this paper. This package is available on the CRAN
repository, as well as 
\url{https://github.com/selective-inference/R-software}.  A Python
implementation is also available, at
\url{https://github.com/selective-inference/Python-software}. 

A highly nontrivial and important question is to figure out how to
combine p-values, such as those in Table \ref{tab:prostate}, to build
a rigorous stopping rule, i.e., a model selection rule.  While we 
recognize its importance, this topic is {\it not the focus} of our paper.
Our focus is to provide a method for computing proper p-values like
those in Table \ref{tab:prostate} in the first place, which we view as
a major step in the direction of answering the model selection
problem in a practically and theoretically satisfactory manner.
Our future work is geared more toward model selection; we also
discuss this problem in more detail in Section \ref{sec:model_sel}.

\subsection{Related work}

There is much recent work on inference for high-dimensional 
regression models.  One class of techniques, e.g., by
\citet{screenclean,stabselect,minnier2011} is based on
sample-splitting or resampling methods. 
Another class of approaches, e.g., by \citet{zhangconf,
  buhlsignif, vdgsignif, montahypo2, montahypo1} is based on
``debiasing'' or ``denoising'' a regularized regression estimator,
like the lasso. The inferential targets considered in the
aforementioned works are all fixed, and not post-selected, like the
targets we study  here.  As we see it, it is clear (at least
conceptually) how to use sample-splitting techniques to accommodate  
post-selection inferential goals; it is much less clear how to do so
with the debiasing tools mentioned above. 

\citet{Berk2013} carry out valid post-selection inference (PoSI)
by considering all possible model selection  
procedures that could have produced the given submodel. As the authors
state, the inferences are generally conservative for particular
selection procedures, but have the advantage that they do not depend
on the correctness of the selected submodel. This same advantage is 
shared by the tests we propose here.  Comparisons of our 
tests, built for specific selection mechanisms, and the PoSI tests,
which are much more general, would be interesting to pursue in 
future work.

\citet{exactlasso}, reporting on work concurrent with that of this
paper, construct p-values and intervals for lasso coefficients at a
fixed value of the regularization parameter $\lambda$ (instead of a
fixed number of steps $k$ along the lasso path, as we consider in
Section \ref{sec:seq_tests}).  This paper and ours both leverage the
same core    
statistical framework, using truncated Gaussian (TG) distributions,
for exact post-selection inference, but differ in the
applications pursued with this framework.  After our work was
completed, there was further progress on the application and
development of exact post-selection inference tools, e.g., by 
\citet{exactscreen},  \citet{exactmeans}, \citet{exactstep},
\citet{exactpca}, \citet{optimalinf}. 

\subsection{Notation and outline}

Our notation in the coming sections is as follows.
For a matrix $M \in \R^{n\times p}$ and list $S=[s_1,\ldots s_r]
\subseteq [1,\ldots p]$, we write $M_S \in \R^{n\times |S|}$ for the 
submatrix formed by extracting the corresponding columns of $M$
(in the specified order).  Similarly for a vector $x \in \R^p$, we
write $x_S$ to denote the relevant subvector.  We write $(M^T M)^+$
for the (Moore-Penrose) pseudoinverse of the square matrix $M^T M$,
and $M^+=(M^T M)^+ M^T$ for the pseudoinverse of the rectangular
matrix $M$.  Lastly, we use $P_L$ for the projection operator onto a
linear space $L$. 

Here is an outline for the rest of this paper.
Section \ref{sec:summary} gives an overview of our main results.  
Section \ref{sec:poly} describes our general framework for exact
conditional inference, with truncated Gaussian (TG) test statistics.
Section \ref{sec:seq_tests} presents applications of this framework 
to three sequential regression procedures: FS, LAR, and lasso.
Section \ref{sec:spacing} derives a key approximation to our 
TG test for LAR, named the {\it spacing test}, which is
considerably simpler (both in terms of form and computational
requirements) than its exact counterpart.  Section \ref{sec:examples}
covers empirical 
examples, and Section \ref{sec:covtest} draws connections between the
spacing and covariance tests.  We finish with a discussion in
Section \ref{sec:discussion}.

\section{Summary of results}
\label{sec:summary}

We now summarize our conditional testing framework, that yield the 
p-values demonstrated in the prostate cancer data example, beginning 
briefly with the general problem setting we consider.  Consider
testing the hypothesis     
\begin{equation}
\label{eq:H0}
H_0: \;\, v^T \theta = 0,
\end{equation}
conditional on having observed $y\in\cP$,
where $\cP$ is a given polyhedral set, and $v$ is a given 
contrast vector.  We derive a test statistic $T(y,\cP, v)$ with
the property that  
\begin{equation}
\label{eq:T0}
T(y,\cP, v) \overset{\P_0}{\sim} \mathrm{Unif}(0,1),
\end{equation}
where $\P_0(\cdot) = \P_{v^T \theta = 0}(\, \cdot \, | \, y \in
\cP)$, 
the probability measure under $\theta$ for which $v^T \theta = 0$,
conditional on $y \in \cP$.  The assertion is that
$T(y,\cP, v)$ is
exactly uniform under the null measure, for any finite $n$ and $p$.
This statement assumes nothing about the polyhedron $\cP$, and
requires only Gaussian errors in the model  
\eqref{eq:mod}.  As it has a uniform null distribution, the test
statistic in \eqref{eq:T0}  serves as its own p-value, and so
hereafter we will refer to it in both ways (test statistic and
p-value). 

Why should we concern ourselves with an event $y \in \cP$, for a
polyhedron $\cP$?  The short answer: for many regression
procedures of interest---in particular, for the sequential algorithms  
FS, LAR, and lasso---the event that the procedure selects a given
model (after a given number of steps) can be represented in this form.
For example, consider FS after one step, with $p=3$ variables total:
the FS procedure selects variable 3, and assigns it a positive
coefficient, if and only if
\begin{align*}
X_3^T y / \|X_3\|_2 &\geq \pm X_1^T y / \|X_1\|_2, \\
X_3^T y / \|X_3\|_2 &\geq \pm X_2^T y / \|X_2\|_2. 
\end{align*}
With $X$ considered fixed, these inequalities can be compactly
represented as $\Gamma y \geq 0$, where the inequality is meant to be
interpreted componentwise, and $\Gamma \in \R^{4 \times n}$ is a
matrix with rows 
$X_3/\|X_3\|_2 \pm X_1/\|X_1\|_2$,
$X_3/\|X_3\|_2 \pm X_2/\|X_2\|_2$.
Hence if \smash{$\hat{j}_1(y)$} and \smash{$\hat{s}_1(y)$} denote the
variable and sign selected by FS at the first step, then we have shown
that  
\begin{equation*}
\Big\{ y : \hat{j}_1(y)=3, \, \hat{s}_1(y)=1 \Big\} = 
\{y : \Gamma y \geq 0\},
\end{equation*}
for a particular matrix $\Gamma$.  The right-hand side above is
clearly a polyhedron (in fact, it is a cone).  To test the significance
of the 3rd variable, conditional on it being selected at the first
step of FS, we consider the null hypothesis 
$H_0$ as in \eqref{eq:H0}, with $v=X_3$, and 
$\cP=\{y : \Gamma y \geq 0\}$.  The test statistic that we
construct in \eqref{eq:T0} is conditionally uniform under the null. 
This can be reexpressed as    
\begin{equation}
\label{eq:T03}
\P_{X_3^T \theta=0} \Big(T_1 \leq \alpha \,\Big|\, 
\hat{j}_1(y)=3, \, \hat{s}_1(y)=1\Big) = \alpha,
\end{equation}
for all $0 \leq \alpha \leq 1$.  The conditioning in \eqref{eq:T03}
is important because it properly accounts for the adaptive (i.e.,
greedy) nature of FS. Loosely speaking, it measures the magnitude of
the linear function $v_3^T y$---not among all $y$ marginally---but
among the vectors $y$ that would result in FS selecting variable 3,
and assigning it a positive coefficient.    

A similar construction holds for a general step $k$ of FS:  
letting \smash{$\hat{A}_k(y)=[\hat{j}_1(y),\ldots \hat{j}_k(y)]$}
denote the active list after $k$ steps (so that FS selects these
variables in this order) and
\smash{$\hat{s}_{A_k}(y)=[\hat{s}_1(y),\ldots \hat{s}_k(y)]$} denote
the signs of the corresponding coefficients, we have, for any fixed   
$A_k$ and \smash{$s_{A_k}$},   
\begin{equation*}
\Big\{y : \hat{A}_k(y) = A_k, \, \hat{s}_{A_k}(y) = s_{A_k}
\Big\} = \{y : \Gamma y \geq 0\},
\end{equation*}
for another matrix $\Gamma$.  With \smash{$v = (X_{A_k}^+)^T e_k$},
where $e_k$ is the $k$th standard basis vector, the hypothesis in
\eqref{eq:H0} is \smash{$e_k^T X_{A_k}^+ \theta = 0$}, 
i.e., it specifies that the last partial
regression coefficient is not significant, in a projected linear model
of $\theta$ on \smash{$X_{A_k}$}.
For $\cP=\{y : \Gamma y \geq 0\}$, the test  
statistic in \eqref{eq:T0} has the property 
\begin{equation}
\label{eq:T0k}
\P_{e_k^T X_{A_k}^+ \theta = 0}
\Big( T_k \leq \alpha \, \Big| \,
\hat{A}_k(y) = A_k, \, \hat{s}_{A_k}(y) = s_{A_k} \Big) = \alpha, 
\end{equation}
for all $0 \leq \alpha \leq 1$.  We emphasize that the p-value in
\eqref{eq:T0k} is exactly (conditionally) uniform under the null, in
finite samples. This is true without placing any restrictions on
$X$ (besides a general position assumption), and notably, without 
assuming linearity of the underlying model (i.e., without assuming
$\theta=X\beta^*$).  
Further, though we described the case for FS
here, essentially the same story holds for LAR and lasso.  The TG
p-values for FS and LAR in Table \ref{tab:prostate} correspond to
tests of hypotheses as in \eqref{eq:T0k}, i.e., tests of  
\smash{$e_k^T X_{A_k}^+ \theta = 0$}, over steps of these procedures.  

An important point to keep in mind throughout is that our testing
framework for the sequential FS, LAR, and lasso procedures is
not specific to the choice \smash{$v = (X_{A_k}^+)^T e_k$}, and
allows for the testing of arbitrary linear contrasts $v^T \theta$ 
(as long as $v$ is fixed by the conditioning event). For concreteness,
we will pay close attention to the case 
\smash{$v = (X_{A_k}^+)^T e_k$}, since it gives us a test for the
significance of variables as they enter the model, but many other
choices of $v$ could be interesting and useful.  

\subsection{Conditional confidence intervals}

A strength of our framework is that our test statistics
can be inverted to make coverage statements about 
arbitrary linear contrasts of $\theta$.
In particular, consider the hypothesis test defined by
\smash{$v= (X_{A_k}^+)^T e_k$}, for the $k$th step of FS  
(similar results apply to LAR and lasso). By inverting our
test statistic in \eqref{eq:T0k}, we obtain a conditional confidence 
interval $I_k$ satisfying  
\begin{equation}
\label{eq:cci}
\P \Big( e_k^T X_{A_k}^+ \theta \in I_k \, \Big| \, 
\hat{A}_k(y) = A_k, \, \hat{s}_{A_k}(y) = s_{A_k} \Big) = 1-\alpha.
\end{equation}
In words, the random interval $I_k$ traps with probability
$1-\alpha$ the coefficient of the last selected variable, in a
regression model that projects $\theta$ onto $X_{A_k}$, conditional
on FS having 
selected variables $A_k$ with signs $s_{A_k}$, after $k$ steps of the
algorithm. As \eqref{eq:cci} is true conditional on 
$\Gamma y \geq 0$, we can also marginalize this statement to yield  
\begin{equation}
\label{eq:si}
\P \Big( e_k^T X_{\hat{A}_k}^+ \theta \in I_k \Big)
= 1-\alpha.
\end{equation}
Note that \smash{$\hat{A}_k=\hat{A}_k(y)$} denotes the random active 
list after $k$ FS steps. Written in the unconditional form
\eqref{eq:si}, we call $I_k$ a {\it selection interval} for the
random quantity \smash{$e_k^T X_{\hat{A}_k}^+ \theta$}.  We use this
name to emphasize the difference in interpretation here, versus the
conditional case: the selection interval covers a {\it moving target},
as both the identity of the $k$th selected variable, and the
identities of all the previously selected variables (which play a role 
in the $k$th partial regression coefficient of $\theta$ on 
\smash{$X_{\hat{A}_k}$}), are random---they depend on $y$.  
 
We have seen that our intervals can be interpreted
conditionally, as in \eqref{eq:cci}, or unconditionally, as in
\eqref{eq:si}.  The  
former is perhaps more aligned with the spirit of
post-selection inference, as it guarantees coverage, conditional 
on the output of our selection procedure.  But the latter
interpretation is also interesting, and in a way,
cleaner. From the unconditional point of view, 
we can roughly think of the selection interval $I_k$ as covering  
the project population coefficient of the ``$k$th most important
variable'' as deemed by  the sequential regression procedure at hand
(FS, LAR, or lasso). Figure \ref{fig:prostate} displays 90\%
confidence intervals at each step of FS, run on the prostate cancer
data set discussed in the introduction.   

\begin{figure}[htb]
\centering
\includegraphics[width=0.75\textwidth]{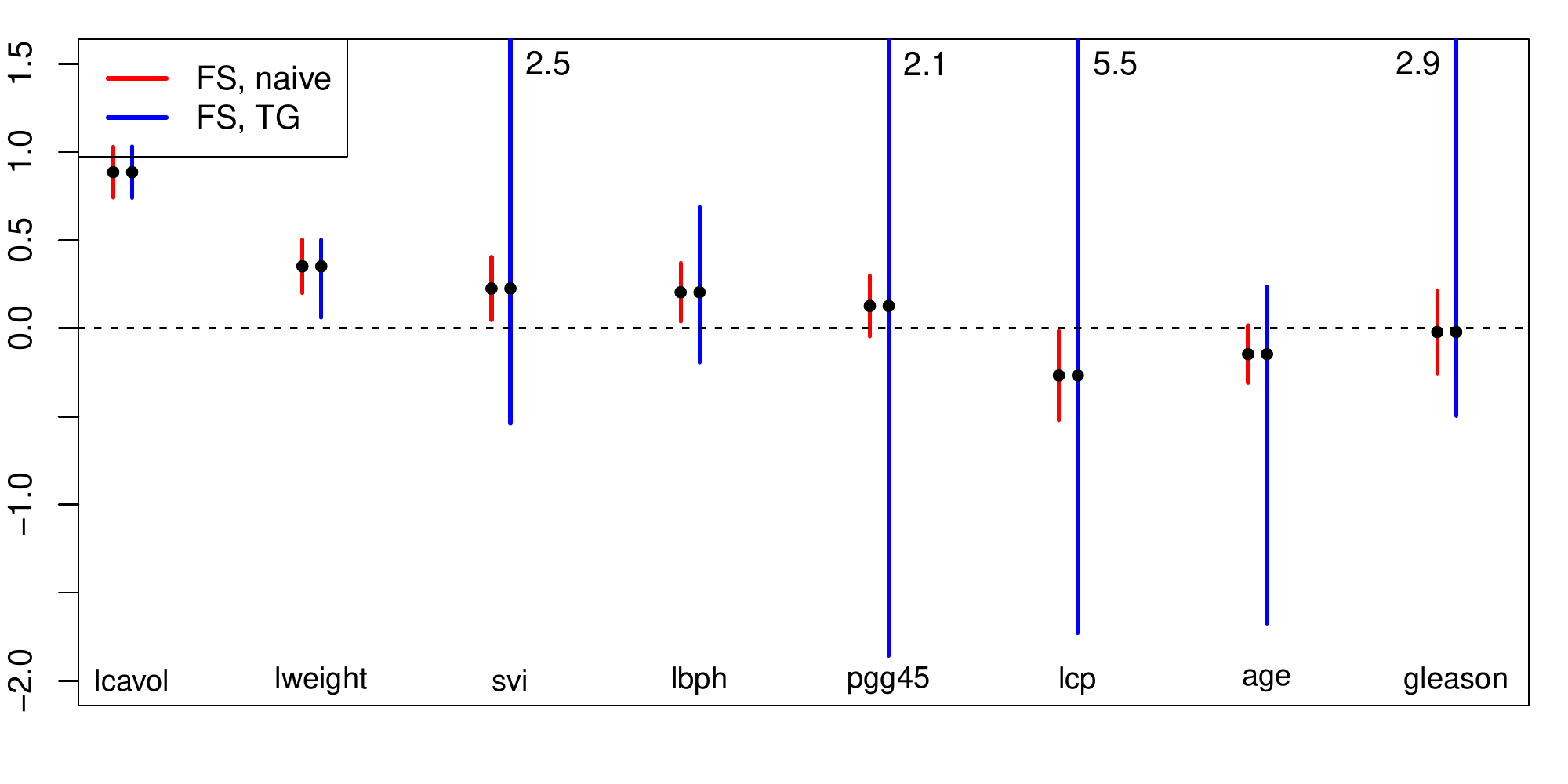}
\caption{\it Prostate cancer data example: naive confidence  
  intervals and 90\% conditional confidence intervals
  (or, selection intervals) computed using the TG (truncated Gaussian)
  statistics, for FS (forward stepwise). Black dots denote
  the estimated partial regression coefficients for the variable to
  enter, in a regression on the active submodel. The upper confidence
  limits for some parameters exceed the range for the y-axis on the
  plot, and their actual values marked at the appropriate places.} 
\label{fig:prostate}
\end{figure}

\subsection{Marginalization}
\label{sec:margin}

Similar to the formation of selection intervals in the last
subsection, we note that any amount of coarsening, i.e.,
marginalization, of the conditioning set in \eqref{eq:T0k} results in
a valid interpretation for p-values.  For example, by marginalizing
over all possible sign lists $s_{A_k}$ associated with $A_k$, we
obtain   
\begin{equation*}
\P_{e_k^T X_{A_k}^+ \theta = 0}
\Big( T_k \leq \alpha \, \Big| \,
\hat{A}_k(y) = A_k \Big) = \alpha, 
\end{equation*}
so that the conditioning event only encodes the observed active   
list, and not the observed signs.  Thus we have another
possible interpretation for the statistic (p-value) $T_k$: under the
null measure, which conditions on FS having selected the
variables $A_k$ (regardless of their signs), $T_k$ is 
uniformly distributed.  The idea of marginalization will be important
when we discuss details of the constructed tests for LAR and
lasso.   

\subsection{Model selection}
\label{sec:model_sel}

How can the inference tools of this paper be translated into rigorous
rules for model selection?  This is of course an important (and
difficult) question, and we do not yet possess a complete
understanding of the model 
selection problem, though it is the topic of future work.  Below we
describe three possible strategies for model selection, using 
the p-values that come from our inference framework.  We do not have 
extensive theory to explain or evaluate them, 
but all are implemented in the R package {\tt selectiveInference}.  

\begin{itemize}
\item {\it Inference from sequential p-values.}  We have advocated the
  idea of computing p-values across steps of the 
  regression procedure at hand, as exemplified in Table
  \ref{tab:prostate}.  Here at each step $k$, the p-value
  tests \smash{$e_k^T X_{A_k}^+ \theta = 0$},
  i.e., tests the significance of the variable to enter the active
  set $A_k$, in a projected linear model of the mean $\theta$ on the
  variables in $A_k$.  \citet{fdrlasso} propose 
  sequential stopping rules using such p-values, including the
  ``ForwardStop'' rule, which guarantees false discovery rate (FDR)
  control at a given level.  For example, the ForwardStop rule at a
  nominal 
  10\% FDR level, applied to the TG p-values from the LAR path
  for the prostate cancer data (the last column of Table
  \ref{tab:prostate}), yields a model with 3 predictors. 
  However, it should be noted that the guarantee for FDR control for
  ForwardStop in \citet{fdrlasso} assumes that the p-values are
  independent, and this is not true for the p-values from our
  inference framework.  

\item {\it Inference at a fixed step $k$.}  Instead of looking at
  p-values across steps, we could instead fix a step $k$, and
  inspect the p-values corresponding to the hypotheses \smash{$e_j^T  
    X_{A_j}^+ \theta = 0$}, for $j=1,\ldots k$. This tests the 
  significance of every variable, among the rest in the discovered
  active set $A_k$, and it still fits within our developed
  framework: we are just utilizing different linear contrasts
  \smash{$v= (X_{A_j}^+)^T e_j$} of the mean $\theta$, for 
  $j=1,\ldots k$. The results of these tests are
  genuinely different, in terms of their statistical meaning, than the
  results from testing variables as they enter the model (since the 
  active set changes at each step).  Given the p-values corresponding
  to all active variables at a given step $k$, we could, e.g., perform
  a Bonferroni correction, and declare significance at the level
  $\alpha/k$, in order to select a model (a subset of $A_k$) with type
  I error controlled at the level $\alpha$.  For example, when we
  apply this strategy at step $k=5$ of the LAR path for the prostate   
  cancer data, and examine Bonferroni corrected p-values at the 0.05  
  level, only two predictors (lweight and pgg45) end up being
  significant.  

\item {\it Inference at an adaptively selected step $k$.}  Lastly,
  the above scheme for inference could be conducted with a step 
  number $k$ that is adaptively selected, instead of fixed ahead of
  time, provided the selection event that determines $k$ is a  
  polyhedral set in $y$.  A specific example of this is an AIC-style  
  rule, which chooses the step $k$ after which the AIC criterion
  rises, say, twice in a row.  We omit the details, but verifying that
  such a stopping rule defines a polyhedral constraint for $y$
  is straightforward (it follows essentially the same logic as the 
  arguments that show the FS selection event is itself
  polyhedral, which are given in Section \ref{sec:poly_fs}).  Hence, by
  including all the necessary polyhedral constraints---those that
  determine $k$, and those that subsequently determine the selected
  model---we can compute p-values for each of the active variables at
  an adaptively selected step $k$, using the inference tools derived
  in this paper.  When this method is applied to the prostate cancer
  data set, the AIC-style rule (which stops once it sees two
  consecutive rises in the AIC criterion) chooses $k=4$.  Examining
  Bonferroni corrected p-values at step $k=4$, only one predictor
  (lweight) remains significant at the 0.05 level.
\end{itemize}

\section{Conditional Gaussian inference after polyhedral selection}  
\label{sec:poly}

In  this section, we present a few key results on Gaussian contrasts 
conditional on polyhedral events, which provides a basis for the
methods proposed in this paper.  The same core development appears in
\citet{exactlasso}; for brevity, we refer the reader to the latter
paper for formal proofs.  We assume
$y \sim N(\theta,\Sigma)$,
where $\theta \in \R^n$ is unknown, but $\Sigma \in \R^{n\times n}$ 
is known. This generalizes our setup in \eqref{eq:mod} 
(allowing for a general error covariance matrix).  We also 
consider a generic polyhedron $\cP = \{y : \Gamma y \geq u\}$, 
where $\Gamma \in \R^{m\times n}$ and $u \in \R^m$ are fixed, and the  
inequality is to be interpreted componentwise.  
For a fixed $v \in \R^n$, our goal is to make inferences about $v^T
\theta$ conditional on $y \in \cP$. Next, we provide a helpful
alternate representation for $\cP$. 

\begin{lemma}[\textbf{Polyhedral selection as truncation}]
\label{lem:truncate}
For any $\Sigma,v$ such that $v^T \Sigma v \not= 0$,
\begin{equation}
\label{eq:truncate}
\Gamma y \geq u \iff
\V^\mathrm{lo}(y) \leq v^T y \leq \V^\mathrm{up}(y), 
\, \V^0(y) \leq 0,
\end{equation}
where 
\begin{align}
\label{eq:v_lower}
\V^\mathrm{lo}(y) &= \max_{j: \rho_j > 0} \,
\frac{u_j  - (\Gamma y)_j + \rho_jv^T y}{\rho_j}, \\ 
\label{eq:v_upper} 
\V^\mathrm{up}(y) &= \min_{j: \rho_j < 0} \,
\frac{ u_j - (\Gamma y)_j + \rho_jv^T y}{\rho_j}, \\
\label{eq:v_zero}
\V^0(y) &= \max_{j: \rho_j = 0} \,
 u_j - (\Gamma y)_j,
\end{align}
and $\rho = \Gamma\Sigma v / v^T\Sigma v$.
Moreover, the triplet $(\V^\mathrm{lo},\V^\mathrm{up},\V^0)(y)$ is
independent of $v^T y$.    
\end{lemma}

\begin{remark} 
The result in \eqref{eq:truncate}, with 
$\V^\mathrm{lo},\V^\mathrm{up},\V^0$ defined as in
\eqref{eq:v_lower}--\eqref{eq:v_zero}, is a deterministic result that  
holds for all $y$. Only the last independence result depends on
normality of $y$.  
\end{remark}

See Figure \ref{fig:lee} for a geometric illustration of this lemma.
Intuitively, we can explain the result as follows, assuming for
simplicity (and without a loss of generality) that $\Sigma=I$.
We first decompose $y=P_v y + P_{v^\perp} y$, where 
$P_v y = vv^T y /\|v\|^2_2$
is the projection of $y$ along $v$, and $P_{v^\perp} y=y-P_vy$ is the
projection onto the orthocomplement of $v$.  Accordingly, we view $y$  
as a deviation from $P_{v^\perp} y$, of an amount $v^T y$, along the line
determined by $v$.  The quantities $\V^\mathrm{lo}$ and
$\V^\mathrm{up}$ describe how far we can deviate on either side of
$P_{v^\perp} y$, before $y$ leaves the polyhedron.  This gives rise to
the inequality $\V^\mathrm{lo} \leq v^T y \leq \V^\mathrm{up}$.  Some
faces of the polyhedron, however, may be perfectly aligned with $v$
(i.e., their normal vectors may be orthogonal to $v$), and $\V^0$
accounts for this by checking that $y$ lies on the correct side of
these faces.  
    
\begin{figure}[htb]
\centering
\includegraphics[width=0.5\textwidth]{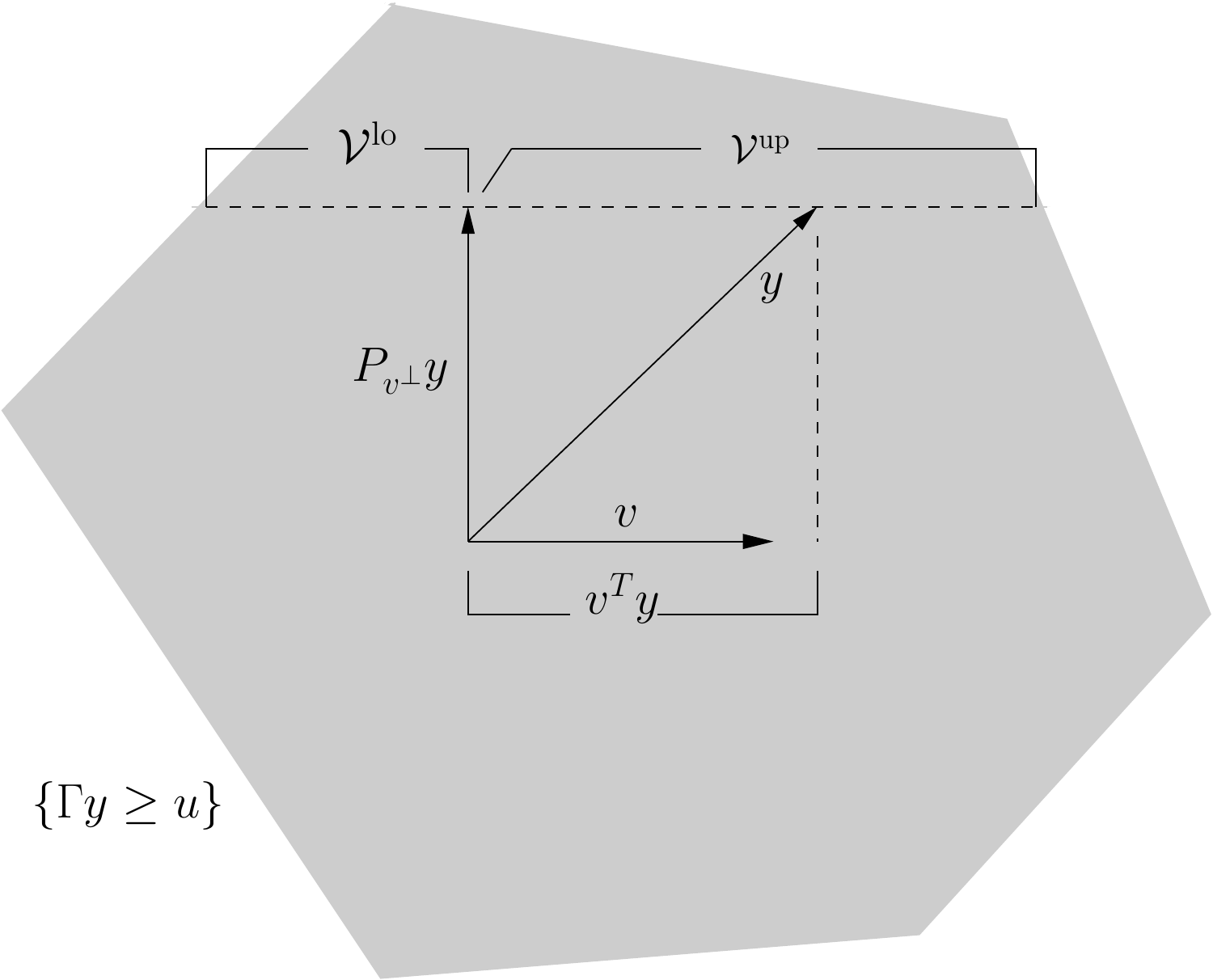}  
\caption{\it Geometry of polyhedral selection as 
  truncation. For simplicity, we
  assume that $\Sigma=I$ (otherwise standardize as
  appropriate).  The shaded gray area is the polyhedral set 
  $\{y : \Gamma y \geq u\} $.  By breaking up $y$ into its projection
  onto $v$ and its projection onto the orthogonal complement of $v$,
  we see that $\Gamma y \geq u$ holds if and only if $v^T y$ does not
  deviate too far from $P_{v^\perp} y$, hence trapping it in between
  bounds $\V^\mathrm{lo},\V^\mathrm{up}$.  Furthermore, these bounds
  $\V^\mathrm{lo},\V^\mathrm{up}$ are functions of $P_{v^\perp} y$
  alone, so under normality, they are independent of $v^T y$.}
\label{fig:lee}
\end{figure}

From Lemma \ref{lem:truncate}, the distribution of 
any linear function $v^Ty$, conditional on the selection 
$\Gamma y \geq u$,  
can be written as the conditional distribution 
\begin{equation}
\label{eq:truncnorm}
v^Ty \, \big|\, \V^\mathrm{lo}(y) \leq v^T y \leq
\V^\mathrm{up}(y), \, \V^0(y) \leq 0.
\end{equation}
Since $v^T y$ has a Gaussian distribution, the above
is a truncated Gaussian distribution (with random
truncation limits).  A simple transformation leads to 
a pivotal statistic, which will be critical for inference about  
$v^T \theta$.   

\begin{lemma}[\textbf{Pivotal statistic after polyhedral selection}]
\label{lem:pivot}
Let $\Phi(x)$ denote the standard normal cumulative distribution
function (CDF), and let \smash{$F_{\mu, \sigma^2}^{[a, b]}$} denote
the CDF of a $N(\mu, \sigma^2)$ random variable truncated to lie in
$[a,b]$, i.e.,  
\begin{equation*}
F_{\mu,\sigma^2}^{[a, b]}(x) = 
\frac{\Phi((x-\mu)/\sigma)  -\Phi((a-\mu)/\sigma)}
{\Phi((b-\mu)/\sigma) - \Phi((a-\mu)/\sigma)}.
\end{equation*}
For $v^T \Sigma v \not= 0$, the statistic 
\smash{$F_{v^T\theta, v^T \Sigma v}^
{[\V^\mathrm{lo}, \V^\mathrm{up}]} (v^T y)$} is a pivotal quantity
conditional on $\Gamma y \geq u$:
\begin{equation}
\label{eq:pivot}
\P\Big( F_{v^T\theta, v^T \Sigma v}
^{[\V^\mathrm{lo}, \V^\mathrm{up}]} (v^T y)
\leq \alpha \,\Big|\, \Gamma y \geq u \Big) = \alpha,
\end{equation}
for any $0 \leq \alpha \leq 1$, 
where $\V^\mathrm{lo}$, $\V^\mathrm{up}$ are as defined in
\eqref{eq:v_lower}, \eqref{eq:v_upper}.
\end{lemma}

\begin{remark}
A referee of this paper astutely noted the connection
between Lemma \ref{lem:pivot} and classic results on inference in 
an exponential 
family model (e.g., Chapter 4 of \citet{tsh}), in the presence of
nuisance parameters.  The analogy is, in a
rotated coordinate system, the parameter of interest is $v^T \theta$,
and the nuisance parameters correspond to $P_v^\perp \theta$.  This
connection is developed in \citet{optimalinf}.
\end{remark}

The pivotal statistic in the lemma leads to valid conditional p-values
for testing the null hypothesis $H_0: v^T \theta =0$, and
correspondingly, conditional confidence intervals for $v^T \theta$.
We divide our presentation into two parts, on one-sided and two-sided
inference.   

\subsection{One-sided conditional inference}

The result below is a direct consequence of the pivot in Lemma
\ref{lem:pivot}.   

\begin{lemma}[\textbf{One-sided conditional inference 
    after polyhedral selection}]  
\label{lem:condinf1}
Given $v^T \Sigma v \not= 0$, suppose that we are interested in
testing  
\begin{equation*}
H_0 : \;\, v^T \theta = 0 \;\;\; \text{against} \;\;\;
H_1 : \;\, v^T \theta > 0.
\end{equation*}
Define the test statistic
\begin{equation}
\label{eq:T1}
T = 1-F_{0, v^T \Sigma v}^
{[\V^\mathrm{lo}, \V^\mathrm{up}]} (v^T y),
\end{equation}
where we use the notation of Lemma \ref{lem:pivot} for the 
truncated normal CDF.
Then $T$ is a valid p-value for $H_0$, conditional on 
$\Gamma y \geq u$: 
\begin{equation}
\label{eq:T1null}
\P_{v^T \theta =0} (T \leq \alpha \, | \, \Gamma y \geq u ) = \alpha,
\end{equation}
for any $0 \leq \alpha \leq 1$.  Further, define $\delta_\alpha$ to
satisfy 
\begin{equation}
\label{eq:delta1}
1- F_{\delta_\alpha, v^T \Sigma v}
^{[\V^\mathrm{lo}, \V^\mathrm{up}]} (v^T y) = \alpha.
\end{equation}
Then $I=[\delta_\alpha,\infty)$ is a valid one-sided confidence
interval for $v^T \theta$, conditional on $\Gamma y \geq u$:
\begin{equation}
\label{eq:coverage1}
\P ( v^T \theta \geq \delta_\alpha 
\, | \, \Gamma y \geq u) = 1 - \alpha. 
\end{equation}
\end{lemma}

Note that by defining our test statistic in terms
of the conditional survival function, as in \eqref{eq:T1},  
we are implicitly aligning ourselves to have power against the
one-sided alternative $H_1 : v^T \theta > 0$.  This is because the
truncated normal survival function
\smash{$1-F_{\mu,\sigma^2}^{[a,b]}(x)$}, evaluated at any fixed point 
$x$, is monotone increasing in $\mu$.  The same fact
(monotonicity of the survival function in $\mu$) 
validates the coverage of the constructed confidence interval in 
\eqref{eq:delta1}, \eqref{eq:coverage1}. 

\subsection{Two-sided conditional inference}

For a two-sided alternative, we use a simple modification
of the one-sided test in Lemma \ref{lem:condinf1}.

\begin{lemma}
[\textbf{Two-sided conditional inference after polyhedral selection}]    
\label{lem:condinf2} 
Given $v^T \Sigma v \not= 0$, suppose that we are interested in
testing  
\begin{equation*}
H_0 : \;\, v^T \theta = 0 \;\;\; \text{against} \;\;\;
H_1 : \;\, v^T \theta \not= 0.
\end{equation*}
Define the test statistic
\begin{equation}
\label{eq:T2}
T = 2 \cdot \min\Big\{F_{0, v^T \Sigma v}
^{[\V^\mathrm{lo}, \V^\mathrm{up}]} (v^T y), \;
1-F_{0, v^T \Sigma v}
^{[\V^\mathrm{lo}, \V^\mathrm{up}]} (v^T y) \Big\},
\end{equation}
where we use the notation of Lemma \ref{lem:pivot} for the 
truncated normal CDF.  Then $T$ is a valid p-value for
$H_0$, conditional on $\Gamma y \geq u$: 
\begin{equation}
\label{eq:T2null}
\P_{v^T \theta =0} (T \leq \alpha \, | \, \Gamma y \geq u ) = \alpha,
\end{equation}
for any $0 \leq \alpha \leq 1$.  Further,
define $\delta_{\alpha/2},\delta_{1-\alpha/2}$ to satisfy 
\begin{align}
\label{eq:delta2a}
1-F_{\delta_{\alpha/2}, v^T \Sigma v}
^{[\V^\mathrm{lo}, \V^\mathrm{up}]} (v^T y) &= \alpha/2, \\ 
\label{eq:delta2b}
1-F_{\delta_{1-\alpha/2}, v^T \Sigma v}
^{[\V^\mathrm{lo}, \V^\mathrm{up}]} (v^T y) &= 1-\alpha/2.
\end{align}
Then
\begin{equation}
\label{eq:coverage2}
\P ( \delta_{\alpha/2} \leq v^T \theta \leq \delta_{1-\alpha/2} 
\, | \, \Gamma y \geq u) = 1 - \alpha. 
\end{equation} 
\end{lemma}

The test statistic in \eqref{eq:T2}, defined in terms of the 
minimum of the truncated normal CDF and survival function, has power
against the two-sided alternative $H_1 : v^T\theta \not= 0$.  The
proof of its null distribution in \eqref{eq:T2null} follows from the
simple fact that if $U$ is a standard uniform random variable, then 
so is $2 \cdot \min\{U,1-U\}$.  The
construction of the confidence interval in \eqref{eq:delta2a},
\eqref{eq:delta2b}, \eqref{eq:coverage2} again uses the monotonicity
of the truncated normal survival function in the underlying mean
parameter.

\section{Exact selection-adjusted tests for FS, LAR, lasso} 
\label{sec:seq_tests}

Here we apply the tools of Section \ref{sec:poly} to the case of
selection in regression using the forward stepwise (FS), least angle
regression (LAR), or lasso procedures.  We assume that the columns of
$X$ are in general position.  This means that for any 
$k<\min\{n,p\}$, any subset of columns $X_{j_1},\ldots X_{j_k}$, and
any signs $\sigma_1,\ldots \sigma_k \in \{-1,1\}$, the affine span of  
$\sigma_1X_{j_1},\ldots \sigma_kX_{j_k}$ does not contain any of the
remaining columns, up to a sign flip 
(i.e., does not contain any of $\pm X_j$, $j\not=j_1,\ldots j_k$).
One can check that this implies the sequence of FS estimates is
unique.  It also implies that the LAR and lasso paths of estimates are
uniquely determined \citep{lassounique}.
The general position assumption is not at all stringent; e.g., 
if the columns of $X$ are drawn according to a continuous probability
distribution, then they are in general position almost surely. 

Next, we show that the model selection events for FS, LAR, and lasso
can be characterized as polyhedra (indeed, cones) of the form  
$\{y : \Gamma y \geq 0\}$.  After this, we describe the forms of
the exact conditional tests and intervals, as provided by Lemmas
\ref{lem:truncate}--\ref{lem:condinf2}, for these procedures, and
discuss some important practical issues. 

\subsection{Polyhedral sets for FS selection events}
\label{sec:poly_fs}

Recall that FS repeatedly adds the 
predictor to the current active model that most improves the
fit.  After each addition, the active coefficients are recomputed by
least squares regression on the active predictors. This process
ends when all predictors are in the model, or when the residual error
is zero.

Formally, suppose that $A_k=[j_1,\ldots j_k]$ is the list of active
variables selected by FS after $k$ steps, and $s_{A_k}=[s_1,\ldots
s_k]$ denotes their signs upon entering.  That is, at each step $k$,   
the variable $j_k$ and sign $s_k$ satisfy 
\begin{align*}
\RSS \big(y,X_{[j_1,\ldots j_{k-1}, j_k]}\big) 
&\leq \RSS \big (y,X_{[j_1,\ldots j_{k-1},j]}\big) 
\;\;\; \text{for all $j\not=j_1,\ldots j_k$, and} \\ 
s_k &= \sign\big( e_k^T (X_{[j_1,\ldots j_k]})^+ y \big), 
\end{align*}
where $\RSS(y,X_S)$ denotes the residual sum of squares from
regressing $y$ onto $X_S$, for a list of variables $S$.  

The set of all observations vectors $y$ that give active
list $A_k$ and sign list $s_{A_k}$ over $k$ steps, denoted
\begin{equation}
\label{eq:Pk_fs}
\cP = \Big\{y : \hat{A}_k(y) = A_k, \, \hat{s}_{A_k}(y) = s_{A_k}
\Big\},
\end{equation}
is indeed a polyhedron of the form $\cP=\{y : \Gamma y \geq 0\}$. 
The proof of this fact uses induction. The case
when $k=1$ can be seen directly by inspection, as $j_1$ and $s_1$ are
the variable and sign to be chosen by FS if and only if  
\begin{align*}
\Big\|\Big(I-X_{j_1}X_{j_1}^T/\|X_{j_1}\|_2^2\Big) y \Big\|_2^2  
&\leq \Big\|\Big(I-X_jX_j^T/\|X_j\|_2^2\Big) y \Big\|_2^2 
\;\;\;\text{for all $j \not=j_1$, and} \\
s_1 &= \sign(X_{j_1}^T y),
\end{align*}
which is equivalent to
\begin{equation*}
s_1 X_{j_1}^T y / \|X_{j_1}\|_2 \geq \pm X_j^T y/\|X_j\|_2  
\;\;\; \text{for all $j\not=j_1$}.
\end{equation*}
Thus the matrix $\Gamma$ begins with $2(p-1)$ rows of the form  
$s_1 X_{j_1}/\|X_{j_1}\|_2 \pm X_j/\|X_j\|_2$, for $j \not= j_1$.
Now assume the statement is true for $k-1$ steps.  At step $k$, the 
optimality conditions for $j_k,s_k$ can be expressed as
\begin{align*}
\Big\|\Big(I-\tilde{X}_{j_k}\tilde{X}_{j_k}^T/
\|\tilde{X}_{j_k}\|_2^2\Big) r \Big\|_2^2  
&\leq \Big\|\Big(I-\tilde{X}_j\tilde{X}_j^T/
\|\tilde{X}_j\|_2^2\Big) r \Big\|_2^2 
\;\;\;\text{for all $j \not=j_1,\ldots j_k$, and} \\
s_k &= \sign(\tilde{X}_{j_k}^T r),
\end{align*}
where \smash{$\tilde{X}_j$} denotes the residual from regressing $X_j$
onto \smash{$X_{A_{k-1}}$}, and $r$ the residual from regressing $y$
onto \smash{$X_{A_{k-1}}$}.  As in the $k=1$ case, the above is
equivalent to 
\begin{equation*}
s_k \tilde{X}_{j_k}^T r / \|\tilde{X}_{j_k}\|_2 
\geq \pm \tilde{X}_j^T r/\|\tilde{X}_j\|_2 
\;\;\; \text{for all $j\not=j_1,\ldots j_k$},
\end{equation*}
or 
\begin{equation*}
s_k X_{j_k}^T P^\perp_{A_{k-1}} y/ 
\|P^\perp_{A_{k-1}} X_{j_k}\|_2 
\geq \pm X_j^T P^\perp_{A_{k-1}} y/ 
\|P^\perp_{A_{k-1}} X_j\|_2 
\;\;\; \text{for all $j\not=j_1,\ldots j_k$},
\end{equation*}
where \smash{$P^\perp_{A_{k-1}}$} denotes the projection orthogonal to the
column space of \smash{$X_{A_{k-1}}$}. Hence we append $2(p-k)$ rows
to $\Gamma$, of the form 
\smash{$ P^\perp_{A_{k-1}} (s_k X_{j_k} /\|P^\perp_{A_{k-1}}
  X_{j_k}\|_2 \pm X_j/\|P^\perp_{A_{k-1}} X_j\|_2)$}, for  
$j \not= j_1,\ldots j_k$. 
In summary, after $k$ steps, the polyhedral set for the FS
selection event \eqref{eq:Pk_fs} corresponds to a matrix $\Gamma$ with
$2pk - k^2 -k$ rows.\footnote{We have been implicitly
assuming thus far that $k < p$.  If $k=p$ (so that necessarily 
$p \leq n$), then we must add an ``extra'' row to $\Gamma$, this row
being \smash{$P^\perp_{A_{p-1}} s_p X_{j_p}$}, which encodes the sign 
constraint \smash{$s_p X_{j_p}^T P^\perp_{A_{p-1}} y \geq 0$}.  For
$k<p$, this constraint is implicitly encoded due to the constraints of
the form \smash{$s_k X_{j_k}^T P^\perp_{A_{k-1}} y \geq \pm a$} for
some $a$.} 

\subsection{Polyhedral sets for LAR selection events}
\label{sec:poly_lar}

The LAR algorithm \citep{lars} is an iterative method, like FS, that
produces a sequence of nested regression models.  As before, we keep a
list of active variables and signs across steps of the algorithm.
Here is a concise description of the LAR steps. At step $k=1$,
we initialize the active variable and sign list with $A=[j_1]$ and 
$s_{A_1}=[s_1]$, where $j_1,s_1$ satisfy  
\begin{equation}
\label{eq:j1s1}
(j_1,s_1) = \argmax_{j =1,\ldots p, \, s \in \{-1,1\}} \, s X_j^T y. 
\end{equation}
(This is the same selection as made by FS at the first step,
provided that $X$ has columns with unit norm.) 
We also record the first knot 
\begin{equation}
\label{eq:lambda1}
\lambda_1 = s_1 X_{j_1}^T y.
\end{equation}
For a general step $k>1$, we form the list $A_k$ by appending $j_k$
to $A_{k-1}$, and form $s_{A_k}$ by appending $s_k$ to 
$s_{A_{k-1}}$, where $j_k,s_k$ satisfy
\begin{equation}
\label{eq:jksk}
(j_k,s_k) = \argmax_{j\notin {A_{k-1}}, \, s \in\{-1,1\}} \, 
\frac{X_j^T P^\perp_{A_{k-1}} y}
{s - X_j^T (X_{A_{k-1}}^+)^T s_{A_{k-1}}} \cdot
1 \left\{\frac{X_j^T P^\perp_{A_{k-1}} y}
{s - X_j^T (X_{A_{k-1}}^+)^T s_{A_{k-1}}} \leq \lambda_{k-1} \right\}. 
\end{equation}
Above, \smash{$P^\perp_{A_{k-1}}$} is the projection orthogonal to the
column space of \smash{$X_{A_{k-1}}$}, $1\{\cdot\}$ denotes the
indicator function, and $\lambda_{k-1}$ is the knot value from step
$k-1$.  We also record the $k$th knot  
\begin{equation}
\label{eq:lambdak}
\lambda_k = \frac{X_{j_k}^T P^\perp_{A_{k-1}} y}
{s_k - X_{j_k}^T (X_{A_{k-1}}^+)^T s_{A_{k-1}}}.
\end{equation}
The algorithm terminates after the $k$-step model if $k=p$, or if 
$\lambda_{k+1} < 0$. 


LAR is often viewed as ``less greedy'' than FS.  It is also
intimately tied to the lasso, as covered in the next subsection.  
Now, we verify that the LAR selection event
\begin{equation}
\label{eq:Pk_lar}
\cP = \Big\{y : \hat{A}_k(y) = A_k, \, \hat{s}_{A_k}(y) = s_{A_k}, \, 
\hat{S}_\ell(y) = S_\ell, \, \ell = 1,\ldots k \Big\}
\end{equation}
is a polyhedron of the form $\cP=\{y : \Gamma y \geq 0\}$.  We can see  
that the LAR event in \eqref{eq:Pk_lar} contains ``extra''
conditioning, \smash{$\hat{S}_\ell(y) = S_\ell$}, $\ell = 1,\ldots k$,
when compared to the FS event in \eqref{eq:Pk_fs}.  
Explained in words, $S_\ell\subseteq \{1,\ldots p\} \times \{-1,1\}$
contains the variable-sign pairs that were ``in competition'' to
become the active variable-sign pair step $\ell$.  A
subtlety of LAR: it is not always the case that $S_\ell=A_{\ell-1}^c
\times \{-1,1\}$, since some variable-sign pairs are automatically
excluded from consideration, as
they would have produced a knot value that is too large (larger than
the previous knot $\lambda_{\ell-1}$).  This is reflected by the
indicator function in \eqref{eq:jksk}.  The
characterization in \eqref{eq:Pk_lar} is still 
{\it perfectly viable for inference}, because any conditional
statement over $\cP$ in \eqref{eq:Pk_lar} translates into a valid one
without conditioning on \smash{$\hat{S}_\ell(y)$}, $\ell=1,\ldots k$,
by marginalizing over all possible realizations 
$S_\ell$, $\ell=1,\ldots k$.  (Recall the discussion of
marginalization in Section \ref{sec:margin}.) 

The polyhedral representation for $\cP$ in \eqref{eq:Pk_lar} again
proceeds by induction. Starting with $k=1$, we can express the
optimality of $j_1,s_1$ in \eqref{eq:j1s1} as
\begin{equation*}
c(j_1,s_1)^T y \geq c(j,s)^T y, \;\;\; \text{for all $j\not=j_1$, $s
  \in \{-1,1\}$},
\end{equation*}
where $c(j,s) = sX_j$.  Thus $\Gamma$ has $2(p-1)$ rows, of the 
form $c(j_1,s_1)-c(j,s)$ for $j\not=j_1$, $s\in\{-1,1\}$. (In the
first step, $S_1=\{1,\ldots p\} \times \{-1,1\}$, and we do not
require extra rows of $\Gamma$ to explicitly represent it.)  Further,
suppose that the selection set can be represented in the desired 
manner, after $k-1$ steps. Then the optimality of
$j_k,s_k$ in \eqref{eq:jksk} can be expressed as
\begin{align*}
c(j_k,s_k,A_{k-1},s_{A_{k-1}})^T y &\geq c(j,s,A_{k-1},s_{A_{k-1}})^T y
\;\;\; \text{for all $(j,s) \in S_k \setminus \{(j_k,s_k)\}$}, \\ 
c(j_k,s_k,A_{k-1},s_{A_{k-1}})^T y &\geq 0,
\end{align*}
where \smash{$c(j,s,A_{k-1},s_{A_{k-1}}) = (P^\perp_{A_{k-1}} X_j)/
(s - X_j^T (X_{A_{k-1}}^+)^T s_{A_{k-1}})$}. The set $S_k$ is
characterized by  
\begin{align*}
c(j,s,A_{k-1},s_{A_{k-1}})^T y &\leq
\lambda_{k-1} \;\;\; \text{for $(j,s) \in S_k$}, \\ 
c(j,s,A_{k-1},s_{A_{k-1}})^T y &\geq
\lambda_{k-1} \;\;\;
\text{for $(j,s) \in \big(A_{k-1}^c \times \{-1,1\}\big) 
\setminus S_k$}.  
\end{align*}
Notice that $\lambda_{k-1}=c(j_{k-1},s_{k-1},A_{k-2},s_{A_{k-2}})^T y$ 
is itself a linear function of $y$, by the inductive hypothesis.    
Therefore, the new $\Gamma$ matrix is created by appending  
the following $|S_k|+2(p-k+1)$ rows to the previous matrix:  
$c(j_k,s_k,A_{k-1},s_{A_{k-1}})- c(j,s,A_{k-1},s_{A_{k-1}})$, 
for $(j,s) \in S_k \setminus \{(j_k,s_k)\}$; 
$c(j_k,s_k,A_{k-1},s_{A_{k-1}})$; 
$c(j_{k-1},s_{k-1},A_{k-2},s_{A_{k-2}})-c(j,s,A_{k-1},s_{A_{k-1}})$,
for $(j,s) \in S_k$; 
$c(j,s,A_{k-1},s_{A_{k-1}})-c(j_{k-1},s_{k-1},A_{k-2},s_{A_{k-2}})$  
for $(j,s) \in (A_{k-1}^c \times \{-1,1\}) \setminus S_k$.  
In total, the number of rows of $\Gamma$ at step $k$ of LAR is   
bounded above by 
\smash{$\sum_{\ell=1}^k (|S_\ell|+2(p-\ell+1)) \leq 3pk-3k^2/2+3k/2$}.   

\subsection{Polyhedral sets for lasso selection events}
\label{sec:poly_lasso}

By introducing a step into the LAR algorithm that deletes variables
from the active set if their coefficients pass through zero, the
modified LAR algorithm traces out the lasso regularization path 
\citep{lars}.  To concisely describe this modification, at a step 
$k>1$, denote by \smash{$(j_k^\mathrm{add},s_k^\mathrm{add})$} the 
variable-sign pair to enter the model next, as defined in
\eqref{eq:jksk}, and denote by \smash{$\lambda_k^\mathrm{add}$} the
value of $\lambda$ at which they would enter, as defined in
\eqref{eq:lambdak}.  Now define 
\begin{equation}
\label{eq:jk_del}
j_k^\mathrm{del} = \argmax_{j \in A_{k-1} \setminus \{j_{k-1}\}} \, 
\frac{e_j^T X_{A_{k-1}}^+ y}
{e_j^T (X_{A_{k-1}}^T X_{A_{k-1}})^{-1} s_{A_{k-1}}} 
\cdot 1 \left\{
\frac{e_j^T X_{A_{k-1}}^+ y}
{e_j^T (X_{A_{k-1}}^T X_{A_{k-1}})^{-1} s_{A_{k-1}}} 
\leq \lambda_{k-1} \right\},
\end{equation}
the variable to leave the model next, and
\begin{equation}
\label{eq:lambdak_del}
\lambda_k^\mathrm{del} = 
\frac{e_{j_k^\mathrm{del}}^T X_{A_{k-1}}^+ y} 
{e_{j_k^\mathrm{del}}^T (X_{A_{k-1}}^T X_{A_{k-1}})^{-1} s_{A_{k-1}}},
\end{equation}
the value of $\lambda$ at which it would leave.  The lasso
regularization path is given by executing whichever action---variable
entry, or variable deletion---happens first, when seen from the
perspective of decreasing $\lambda$.  That is, we record the $k$th
knot \smash{$\lambda_k =
\max\{\lambda_k^\mathrm{add},\lambda_k^\mathrm{del}\}$}, and we form 
\smash{$A_k,s_{A_k}$} by either adding
\smash{$j_k^\mathrm{add},s_k^\mathrm{add}$} to
\smash{$A_{k-1},s_{A_{k-1}}$} if 
\smash{$\lambda_k = \lambda_k^\mathrm{add}$}, or by deleting 
\smash{$j_k^\mathrm{del}$} from $A_{k-1}$ and its sign from 
\smash{$s_{A_{k-1}}$} if 
\smash{$\lambda_k = \lambda_k^\mathrm{del}$}.

We show that the lasso selection event\footnote{The observant reader
  might notice that the selection event 
  for the lasso in \eqref{eq:Pk_lasso}, compared to that for FS in
  \eqref{eq:Pk_fs} and LAR in \eqref{eq:Pk_lar}, actually enumerates
  the assignments of active sets \smash{$\hat{A}_\ell(y)=A_\ell$}, 
  $\ell=1,\ldots k$ across all $k$ steps of the path.  This is done
  because, with variable deletions, it is no longer possible to
  express an entire history of active sets with a single list.  The
  same is true of the active signs.}
\begin{equation}
\label{eq:Pk_lasso}
\cP = \Big\{y : \hat{A}_\ell(y) = A_\ell, \, 
\hat{s}_{A_\ell}(y) = s_{A_\ell}, \,  
\hat{S}_\ell^\mathrm{add}(y) = S_\ell^\mathrm{add}, \, 
\hat{S}_\ell^\mathrm{del}(y) = S_\ell^\mathrm{del}, \,
\ell = 1,\ldots k \Big\}.
\end{equation}
can be expressed in polyhedral form 
$\{ y: \Gamma y \geq  0\}$.  A difference between \eqref{eq:Pk_lasso}
and the LAR event in \eqref{eq:Pk_lar} is that, in addition to keeping
track of the set \smash{$S^\mathrm{add}_\ell$} of variable-sign
pairs in consideration to become active (to be added) at step $\ell$,
we must also keep track of the set \smash{$S^\mathrm{del}_\ell$} of 
variables in consideration to become inactive (to be deleted) at step 
$\ell$.  As discussed earlier, a valid inferential statement
conditional on the lasso event $\cP$ in \eqref{eq:Pk_lasso} is still
valid once we ignore the conditioning on
\smash{$\hat{S}_\ell^\mathrm{add}(y)$},
\smash{$\hat{S}_\ell^\mathrm{del}(y)$},
$\ell=1,\ldots k$, by marginalization. 

To build the $\Gamma$ matrix corresponding to \eqref{eq:Pk_lasso}, we
begin the same  
construction as we laid out for LAR in the last subsection,
and simply add more rows.  At a step $k>1$, the rows we
described appending to $\Gamma$ for LAR now
merely characterize the variable-sign pair 
\smash{$(j_k^\mathrm{add},s_k^\mathrm{add})$} to enter the model
next, as well as the set \smash{$S_k^\mathrm{add}$}. To characterize
the variable \smash{$j_k^\mathrm{del}$} to leave the model next, we
express its optimality in \eqref{eq:jk_del} as 
\begin{align*}
d(j_k^\mathrm{del},A_{k-1},s_{A_{k-1}})^T y &\geq
d(j,A_{k-1},s_{A_{k-1}})^T y 
\;\;\; \text{for all $j \in 
S_k^\mathrm{del} \setminus \{j_k^\mathrm{del}\}$}, \\  
d(j_k^\mathrm{del},A_{k-1},s_{A_{k-1}})^T y &\geq 0, 
\end{align*}
where \smash{$d(j,A_{k-1},s_{A_{k-1}}) = ((X_{A_{k-1}}^+)^T e_j) /
(e_j^T (X_{A_{k-1}}^T X_{A_{k-1}})^{-1} s_{A_{k-1}})$}, and 
$S_k^\mathrm{del}$ is characterized by   
\begin{align*}
d(j,s,A_{k-1},s_{A_{k-1}})^T y &\leq
\lambda_{k-1} \;\;\; \text{for $(j,s) \in S_k^\mathrm{del}$}, \\ 
d(j,s,A_{k-1},s_{A_{k-1}})^T y &\geq
\lambda_{k-1} \;\;\;
\text{for $(j,s) \in A_{k-1} \setminus S_k^\mathrm{del}$}. 
\end{align*} 
Recall that $\lambda_{k-1}=b_{k-1}^T y$ is a linear function
of $y$, by the 
inductive hypothesis. If a variable was added at step $k-1$, then 
$b_{k-1}=c(j_{k-1},s_{k-1},A_{k-2},s_{A_{k-2}})$; if instead
a variable was deleted at step $k-1$, then
$b_{k-1}=d(j_{k-1},A_{k-2},s_{A_{k-2}})$.  Lastly, we must
characterize step $k$ as either witnessing a variable addition or
deletion.  The former case is represented by
\begin{equation*}
c(j_k^\mathrm{add},s_k^\mathrm{add},A_{k-1},s_{A_{k-1}})^T \geq 
d(j_k^\mathrm{del},A_{k-1},s_{A_{k-1}})^T y,
\end{equation*}
the latter case reverses the above inequality. Hence, in
addition to those described in the previous subsection, we append the  
following   
\smash{$|S_k^\mathrm{del}|+|A_{k-1}|+1$} rows to $\Gamma$: 
\smash{$d(j_k^\mathrm{del},A_{k-1},s_{A_{k-1}}) -
  d(j,A_{k-1},s_{A_{k-1}})$} for 
\smash{$(j,s)\in S_k^\mathrm{del} \setminus \{j_k^\mathrm{del}\}$}; 
\smash{$d(j_k^\mathrm{del},A_{k-1},s_{A_{k-1}})$}; 
$b_{k-1}-d(j,A_{k-1},s_{A_{k-1}})$
for \smash{$(j,s) \in S_k^\mathrm{del}$};
$d(j,A_{k-1},s_{A_{k-1}})- b_{k-1}$ 
for \smash{$(j,s) \in A_{k-1} \setminus S_k^\mathrm{del}$}; and 
either
\smash{$c(j_k^\mathrm{add},s_k^\mathrm{add},A_{k-1},s_{A_{k-1}}) - 
d(j_k^\mathrm{del},A_{k-1},s_{A_{k-1}})$}, or the negative of this
quantity, depending on whether a variable was added or deleted at step
$k$. Altogether, the number of rows of $\Gamma$ at step $k$ 
is at most \smash{$\sum_{\ell=1}^k
  (|S_\ell^\mathrm{add}|+|S_\ell^\mathrm{del}|+2|A_{\ell-1}^c|+|A_{\ell-1}|+1)
  \leq 3pk+k$}.

\subsection{Details of the exact tests and intervals}
\label{sec:detail_test}

Given a number of steps $k$, after we have formed the appropriate  
$\Gamma$ matrix for the FS, LAR, or lasso procedures, as derived in
the last three subsections, computing conditional $p$-values and
intervals is straightforward. 
Consider testing a generic null hypothesis $H_0 : v^T \theta = 0$
where $v$ is arbitrary. First we compute, as prescribed by Lemma  
\ref{lem:truncate}, the quantities 
\begin{align*}
\V^\mathrm{lo} &= \max_{j : (\Gamma v)_j > 0} \, 
-(\Gamma y)_j \cdot \|v\|_2^2/(\Gamma v)_j + v^T y, \\   
\V^\mathrm{up} &= \min_{j : (\Gamma v)_j < 0} \, 
-(\Gamma y)_j \cdot \|v\|_2^2/(\Gamma v)_j + v^T y.
\end{align*}
Note that the number of operations needed to compute 
\smash{$\V^\mathrm{lo},\V^\mathrm{up}$} is $O(mn)$, where $m$ is the 
number of rows of $\Gamma$.  For testing against a one-sided
alternative $H_1 : v^T \theta > 0$, we form the test statistic  
\begin{equation*}
T_k = 1 - F_{0, \sigma^2 \|v\|_2^2} 
^{[\V^\mathrm{lo}, \V^\mathrm{up}]} (v^T y)
= \frac{\Phi\left(\frac{\V^\mathrm{up}}{\sigma \|v\|_2}\right) - 
\Phi\left(\frac{v^T y}{\sigma\|v\|_2}\right)}
{\Phi\left(\frac{\V^\mathrm{up}}{\sigma \|v\|_2}\right) - 
\Phi\left(\frac{\V^\mathrm{lo}}{\sigma \|v\|_2}\right)}.
\end{equation*}
By Lemma \ref{lem:condinf1}, this serves as valid p-value,
conditional on the selection. That is,   
\begin{equation}
\label{eq:Tvnull}
\P_{v^T \theta =0} \Big(T_k \leq \alpha \,\Big|\, 
\hat{A}_k(y) = A_k, \, \hat{s}_{A_k}(y) = s_{A_k} \Big) = \alpha, 
\end{equation}
for any $0 \leq \alpha \leq 1$.  Also by Lemma
\ref{lem:condinf1}, a conditional confidence interval is derived
by first computing $\delta_\alpha$ that satisfies
\begin{equation*}
1-F_{\delta_\alpha, \sigma^2 \|v\|_2^2}
^{[\V^\mathrm{lo}, \V^\mathrm{up}]} (v^T y) = \alpha.
\end{equation*}
Then we let $I_k = [\delta_\alpha,\infty)$, which has the 
proper conditional coverage, in that
\begin{equation}
\label{eq:Tvcover}
\P \Big( v^T \theta \in I_k \,\Big|\,
\hat{A}_k(y) = A_k, \, \hat{s}_{A_k}(y) = s_{A_k} \Big) = 1-\alpha.
\end{equation}
For testing against a two-sided alternative 
$H_1 : v^T \theta \not= 0$, we instead use the test statistic  
\begin{equation*}
T_k' = 2 \cdot \min\{T_k,1-T_k\},
\end{equation*}
and by Lemma \ref{lem:condinf2}, the same results as in 
\eqref{eq:Tvnull}, \eqref{eq:Tvcover} follow, but with $T_k'$ in
place of $T_k$, and  $I_k'=[\delta_{\alpha/2}, \delta_{1-\alpha/2}]$
in place of $I_k$. 

Recall that the case when \smash{$v= (X_{A_k}^+)^T e_k$}, and the
null hypothesis is $H_0 : e_k^T X_{A_k}^+ \theta = 0$, is of
particular interest, as discussed in Section \ref{sec:summary}.  Here,
we are testing whether the coefficient of the last selected
variable, in the population regression of $\theta$ on $X_{A_k}$, is
equal to zero.   
For this problem, the details of the p-values and intervals follow 
exactly as above with the appropriate substitution for $v$.  However,
as we examine next, the one-sided variant of the test must be handled 
with care, in order for the alternative to make sense. 


\subsection{One-sided or two-sided tests?}
\label{sec:oneor}

Consider testing the partial regression coefficient of the variable to 
enter, at step $k$ of FS, LAR, or lasso, in a projected linear model
of $\theta$ on \smash{$X_{A_k}$}. With the choice  
\smash{$v= (X_{A_k}^+)^T e_k$}, the one-sided setup  
$H_0 : v^T \theta=0$ versus $H_1 : v^T \theta > 0$ is not inherently 
meaningful, since there is no reason to believe ahead of time that
the $k$th population regression coefficient 
\smash{$e_k^T X_{A_k}^+\theta$} should be positive. 
By defining \smash{$v= s_k (X_{A_k}^+)^T e_k$}, where recall $s_k$ is 
the sign of the $k$th variable as it enters the (FS, LAR, or lasso)
model, the null \smash{$H_0 : s_k e_k^T X_{A_k}^+\theta=0$}
is unchanged, but the one-sided alternative 
\smash{$H_1 : s_k e_k^T X_{A_k}^+\theta > 0$} now has a concrete
interpretation: it says that the population regression
coefficient of the last selected variable is nonzero, and {\it has
  the same sign} as the coefficient in the fitted (sample) model. 

Clearly, the one-sided test here will have stronger power than its
two-sided version when the described one-sided alternative is
true.  It will lack power when the appropriate population regression
coefficient is nonzero, and has the opposite sign as the coefficient in
the sample model.  However, this is not really of concern, because the 
latter alternative seems unlikely to be encountered in practice,
unless the size of the population effect is very small (in which case 
the two-sided test would not likely reject, as well).  For
these reasons, we often prefer the one-sided test, 
with \smash{$v= s_k (X_{A_k}^+)^T e_k$}, for pure 
significance testing of the variable to enter at the $k$th step.  The
p-values in Table \ref{tab:prostate}, e.g., were computed accordingly.  

With confidence intervals, the story is different.
Informally, we find one-sided (i.e., half-open) intervals, which
result from a one-sided significance test, to be less desirable from
the perspective of a practitioner.
Hence, for coverage statements, we often prefer the two-sided
version of our test, which leads to two-sided conditional confidence
intervals (selection intervals).  The intervals in Figure
\ref{fig:prostate}, for example, were computed in this way.

\subsection{Models with intercept}

Often, we run FS, LAR, or lasso by first beginning with an intercept 
term in the model, and then adding predictors.  Our
selection theory can accommodate this case.  It is easiest to simply 
consider centering $y$ and the columns of $X$, which is equivalent to
including an intercept term in the regression.  After centering, the  
covariance matrix of $y$ is 
$\Sigma=\sigma^2(I-\mathds{1}\mathds{1}^T/n)$, 
where $\mathds{1}$ is the vector of all $1$s.  This is fine, because
the polyhedral theory from Section \ref{sec:poly} applies to Gaussian
random variables with an arbitrary (but known) covariance.  With the
centered $y$ and $X$, the construction of the polyhedral set 
($\Gamma$ matrix) carries over just as described in Sections
\ref{sec:poly_fs}, \ref{sec:poly_lar}, or \ref{sec:poly_lasso}.
The conditional tests and intervals also carry
over as in Section \ref{sec:detail_test}, except with the general
contrast vector $v$ replaced by its own centered version.  Note that 
when $v$ lies in the column space of $X$, e.g., when 
\smash{$v= (X_{A_k}^+)^T e_k$}, no changes at all are needed.  

\subsection{How much to condition on?}

In Sections \ref{sec:poly_fs}, \ref{sec:poly_lar}, and
\ref{sec:poly_lasso}, we saw in the construction of the polyhedral
sets in \eqref{eq:Pk_fs}, \eqref{eq:Pk_lar}, \eqref{eq:Pk_lasso} that
it was convenient to condition on different quantities in order to
define the FS, LAR, and lasso selection events, respectively.  All
three of the 
polyhedra in \eqref{eq:Pk_fs}, \eqref{eq:Pk_lar}, \eqref{eq:Pk_lasso}
condition on the active signs \smash{$s_{A_k}$} of the selected model,
and the latter two condition on more (loosely, the set of variables
that were eligible to enter or leave the active model at each step).
The decisions here, about what to
condition on, were driven entirely by computational convenience.  
It is important to note that---even though any amount of extra
conditioning will still lead to valid inference once we marginalize 
out part of the conditioning set (recall Section \ref{sec:margin})---a 
greater degree of conditioning will generally lead to less powerful
tests and wider intervals.  This not only refers to the extra
conditioning in the LAR and lasso selection events, but also to the
specification of active signs \smash{$s_{A_k}$} common to all three
events.  At the price of increased computation, one
can eliminate unnecessary conditioning by considering a union of 
polyhedra (rather than a single one) as determining a selection
event.  This is done in \citet{exactlasso} and \citet{exactmeans}.
In FS regression, one can condition only on the sufficient statistics
for the nuisance parameters, and obtain the most powerful selective
test. Details are in \cite{optimalinf, fithian2015}.

\section{The spacing test for LAR}
\label{sec:spacing}

A computational challenge faced by the FS, LAR, and lasso tests
described in the last section is that the matrices
$\Gamma$ computed for the polyhedral representations 
$\{y : \Gamma y \geq 0\}$ of their 
selection events can grow very large; in the FS case, the matrix
$\Gamma$ will have $2pk$ after $k$ steps, and for LAR and lasso, it
will have roughly $3pk$ rows. This makes it cumbersome to form 
$\V^\mathrm{lo},\V^\mathrm{up}$,
as the computational cost for these quantities scales linearly with 
the number of rows of $\Gamma$.  In this section, we derive a simple
approximation to the polyhedral representations for the LAR 
events, which remedies this computational issue.

\subsection{A refined characterization of the polyhedral set}
\label{sec:refined}

We begin with an alternative characterization for the LAR selection
event, after $k$ steps. The proof draws heavily on results from
\citet{LTTT2013}, and is given in Appendix \ref{app:refined}.  

\begin{lemma}
\label{lem:refined}
Suppose that the LAR algorithm produces the list of active variables
$A_k$ and signs $s_{A_k}$ after $k$ steps.
Define \smash{$c(j,s,A_{k-1},s_{A_{k-1}}) =  
(P^\perp_{A_{k-1}} X_j)/ (s - X_j^T (X_{A_{k-1}}^+)^T
s_{A_{k-1}})$}, with the convention $A_0=s_{A_0}=\emptyset$,
so that \smash{$c(j,s,A_0,s_{A_0})=c(j,s)=sX_j$}.  Consider the
following conditions:
\begin{align}
\label{eq:ochar}
c(j_1,s_1,A_0,s_{A_0})^T y &\geq c(j_2,s_2, A_1,s_{A_1})^T y \geq
\ldots \geq c(j_k,s_k, A_{k-1},s_{A_{k-1}})^T y \geq 0, \\
\label{eq:pchar}
 c(j_k,s_k, A_{k-1},s_{A_{k-1}})^T y &\geq M^+_k
\Big(j_k,s_k, c(j_{k-1},s_{k-1}, A_{k-2},s_{A_{k-2}})^T y\Big), \\
\label{eq:nchar}
 c(j_\ell,s_\ell, A_{\ell-1},s_{A_{\ell-1}})^T y &\leq M^-_\ell
\Big(j_\ell,s_\ell, c(j_{\ell-1},s_{\ell-1}, A_{\ell-2},s_{A_{\ell-2}})^T y\Big), \;\;\;
\text{for $\ell=1,\ldots k$}, \\
\label{eq:zchar}
 0 &\geq M^0_\ell
\Big(j_\ell,s_\ell, c(j_{\ell-1},s_{\ell-1}, A_{\ell-2},s_{A_{\ell-2}})^T y\Big), \;\;\;
\text{for $\ell=1,\ldots k$}, \\
\label{eq:schar}
0 &\leq M^S_\ell \, y, \;\;\;\text{for $\ell=1,\ldots k$}.
\end{align}
(Note that for $\ell=1$ in \eqref{eq:nchar}, 
\eqref{eq:zchar}, we are meant to interpret
$c(j_0,s_0,A_{-1},s_{A_{-1}})^T y = \infty$.) The set of all $y$ 
satisfying the above conditions is 
the same as the set $\cP$ in \eqref{eq:Pk_lar}.

Moreover, the quantity $M^+_k$ in \eqref{eq:pchar} can be written as 
a maximum of linear functions of $y$, each $M^-_\ell$ in
\eqref{eq:nchar} can be written as a minimum of linear functions of
$y$, each $M^0_\ell$ in \eqref{eq:zchar} can be written as a
maximum of linear functions of $y$, and each $M^S_\ell$ in
\eqref{eq:schar} is a matrix.  Hence 
\eqref{eq:ochar}--\eqref{eq:schar} can be expressed as $\Gamma y
\geq 0$ for a matrix $\Gamma$. The number of rows of
$\Gamma$ is bounded above by $4pk-2k^2-k$. 
\end{lemma}


At first glance, Lemma \ref{lem:refined} seems to have done little
for us over the 
polyhedral characterization in Section \ref{sec:poly_lar}: after $k$
steps, we are now faced with a $\Gamma$ matrix that has on the order
of $4pk$ rows (even more than before!).  Meanwhile, at the risk of
stating the obvious, the characterization in Lemma \ref{lem:refined}
is far more succinct (i.e., the $\Gamma$ matrix is much smaller)  
without the conditions in \eqref{eq:nchar}--\eqref{eq:schar}.
Indeed, in certain special cases (e.g., orthogonal predictors)
these conditions are vacuous, and so they do not contribute to the 
formation of $\Gamma$. Even outside of such cases,
we have found that dropping the conditions
\eqref{eq:nchar}--\eqref{eq:schar} yields an accurate (and
computationally efficient) approximation of the LAR selection set in
practice. This is discussed next.  

\subsection{A simple approximation of the polyhedral set} 
\label{sec:approx}

It is not hard to see from their definitions in Appendix
\ref{app:refined} that when $X$ is orthogonal (i.e., when $X^T X =
I$), we have $M^-_\ell=\infty$ and $M^0_\ell=-\infty$, 
and furthermore, the matrix $M_\ell^S$ has zero rows, for each 
$\ell$. This means that the conditions
\eqref{eq:nchar}--\eqref{eq:schar} are vacuous.  The 
polyhedral characterization in Lemma \ref{lem:refined}, therefore,
reduces to $\{y : \Gamma y \geq U\}$, where $\Gamma$ has only  
$k+1$ rows, defined by the $k+1$ constraints \eqref{eq:ochar},
\eqref{eq:pchar}, and $U$ is a random vector with components 
$U_1=\ldots = U_k=0$, and
\smash{$U_{k+1}=M^+_k(j_k,s_k, c(j_{k-1},s_{k-1},
  A_{k-2},s_{A_{k-2}})^T y)$}.   

For a general (nonorthogonal) $X$, we might still consider ignoring
the conditions \eqref{eq:nchar}--\eqref{eq:schar} and using the
compact representation $\{y : \Gamma y \geq U\}$ induced by
\eqref{eq:ochar}, \eqref{eq:pchar}. This is an
approximation to the exact polyhedral characterization in Lemma
\ref{lem:refined}, but it is a computationally favorable one, since
$\Gamma$ has only $k+1$ rows (compared to about $4pk$ rows per the
construction of the lemma).   Roughly speaking, the constraints in
\eqref{eq:nchar}--\eqref{eq:schar} are often inactive (loose) among
the full collection \eqref{eq:ochar}--\eqref{eq:schar}, so
dropping them does not change the geometry of the set.  Though 
we do not pursue formal arguments to this end (beyond the orthogonal
case), empirical evidence suggests that this approximation is often
justified. 

Thus let us suppose for the moment that we are interested in the
polyhedron $\{y : \Gamma y \geq U\}$ with $\Gamma, U$ as
defined above, either serving an exact representation, or an
approximate one, reducing the full description in Lemma
\ref{lem:refined}.   
Our focus is the application of our polyhedral inference 
tools from Section \ref{sec:poly} to $\{y : \Gamma y \geq U\}$.
Recall that the established polyhedral theory considers sets of
the form $\{y : \Gamma y \geq u\}$, where $u$ is fixed. As the 
equivalence in \eqref{eq:truncate} is a deterministic rather than a 
distributional result, it holds whether $U$ is random or
fixed.  But the independence of the constructed 
$\V^\mathrm{lo},\V^\mathrm{up},\V^0$ and $v^T y$ is not as immediate.
The quantities $\V^\mathrm{lo},\V^\mathrm{up},\V^0$ are now functions
of $y$ and $U$, both of which are random.  A important special case
occurs when $v^T y$ and the pair $((I - \Sigma v v^T / v^T \Sigma v)
y, U)$ are independent.
In this case $\V^\mathrm{lo},\V^\mathrm{up},\V^0$---which only depend
on the latter pair above---are clearly independent of $v^T y$. To be
explicit, we state this result as a corollary.   

\begin{corollary}
[\textbf{Polyhedral selection as truncation, random $U$}] 
\label{cor:truncate}
For any fixed $y, \Gamma, U, v$ with $v^T \Sigma v \not= 0$,  
\begin{equation*}
\Gamma y \geq U \iff
\V^\mathrm{lo}(y,U) \leq v^T y \leq \V^\mathrm{up}(y, U), \, 
\V^0(y,U) \leq 0,
\end{equation*}
where 
\begin{align*}
\V^\mathrm{lo}(y,U) &= \max_{j: \rho_j > 0} \,
\frac{U_j  - (\Gamma y)_j + \rho_jv^T y}{\rho_j}, \\ 
\V^\mathrm{up}(y,U) &= \min_{j: \rho_j < 0} \,
\frac{ U_j - (\Gamma y)_j + \rho_jv^T y}{\rho_j}, \\
\V^0(y,U) &= \max_{j: \rho_j = 0} \,
 U_j - (\Gamma y)_j,
\end{align*}
and $\rho = \Gamma\Sigma v / v^T\Sigma v$. 
Moreover, assume that $y$ and $U$ are random, and that 
\begin{equation}
\label{eq:ucond}
U \;\,\text{is a function of}\;\, (I-\Sigma v v^T / v^T \Sigma v) y,
\end{equation}
so $v^T y$ and the pair $((I - \Sigma v v^T / v^T \Sigma v)
y, U)$ are independent.
Then the triplet $(\V^\mathrm{lo},\V^\mathrm{up},\V^0)(y,U)$ is
independent of $v^T y$.      
\end{corollary}

Under the condition \eqref{eq:ucond} on $U$, the rest of
the inferential treatment proceeds as before, as Corollary
\ref{cor:truncate} ensures that we have the required alternate
truncated Gaussian representation of $\Gamma y \geq U$, with the
random truncation limits $\V^\mathrm{lo},\V^\mathrm{up}$ being
independent of the univariate Gaussian $v^T y$. 
In our LAR problem setup, $U$ is a given random variate (as
described in the first paragraph of this subsection).  The relevant
question is of course: when does \eqref{eq:ucond}
hold? Fortunately, this condition holds with only very minor 
assumptions on $v$: this vector must lie
in the column space of the LAR active variables at the current step.

\begin{lemma}
\label{lem:indep}
Suppose that we have run $k$ steps of LAR, and represent the
conditions \eqref{eq:ochar}, \eqref{eq:pchar} in Lemma
\ref{lem:refined} as $\Gamma y \geq U$.
Under our  
running regression model $y \sim N(\theta,\sigma^2 I)$, if
$v$ is in the column space of the active variables $A_k$, written
\smash{$v \in \col(X_{A_k})$}, then the     
condition in \eqref{eq:ucond} holds, so inference
for $v^T \theta$ can be carried out with the same set of tools as
developed in Section \ref{sec:poly}, conditional on 
$\Gamma y \geq U$.   
\end{lemma}

The proof is given in Appendix \ref{app:indep}.  For example, 
if we choose the contrast vector to be 
\smash{$v= (X_{A_k}^+)^T e_k$}, a case we have revisited throughout  
the paper, then this satisfies the conditions of Lemma
\ref{lem:indep}. 
Hence, for testing the significance of the projected regression
coefficient of the latest selected LAR variable, conditional on
$\Gamma y \geq U$, we may use the p-values and intervals derived in
Section \ref{sec:poly}.  We walk through this usage in the next
subsection.   

\subsection{The spacing test}

The (approximate) representation of the form 
$\{y : \Gamma y \geq U\}$ derived in the last subsection (where
$\Gamma$ is small, having $k+1$ rows), can only be used to
conduct inference over $v^T \theta$ for certain vectors $v$, namely,
those lying in the span of current active LAR variables. The
particular choice of contrast vector
\begin{equation}
\label{eq:lar_v}
v = c(j_k,s_k,A_{k-1},s_{A_{k-1}}) = 
\frac{P^\perp_{A_{k-1}} X_{j_k}}
{s_k - X_{j_k}^T (X_{A_{k-1}}^+)^T s_{A_{k-1}}},
\end{equation}
paired with the compact representation $\{y : \Gamma y \geq U\}$,
leads to a very special test that we name the {\it spacing test}.
From the definition \eqref{eq:lar_v}, and the well-known formula for
partial regression coefficients, we see that the null hypothesis being
considered is 
\begin{equation*}
H_0 : \;\, v^T \theta = 0 \iff H_0 : \;\, e_k^T X_{A_k}^+ \theta = 0,
\end{equation*}
i.e., the spacing test is a test for the $k$th coefficient in
the regression of $\theta$ on \smash{$X_{A_k}$}, just as we
have investigated all along under the equivalent choice of contrast
vector \smash{$v = (X_{A_k}^+)^T e_k$}. 
The main appeal of the spacing test lies in its simplicity.  Letting
\begin{equation}
\label{eq:weight}
\weight_k = \big\| (X_{A_k}^+)^T s_{A_k} - 
(X_{A_{k-1}}^+)^T s_{A_{k-1}} \big\|_2,
\end{equation}
the spacing test statistic is defined by
\begin{equation}
\label{eq:spacing}
T_k = \frac{\Phi(\lambda_{k-1} \frac{\weight_k}{\sigma}) -
\Phi(\lambda_k \frac{\weight_k}{\sigma})}
{\Phi(\lambda_{k-1} \frac{\weight_k}{\sigma}) -
\Phi(M^+_k \frac{\weight_k}{\sigma})}.
\end{equation}
Above, $\lambda_{k-1}$ and $\lambda_k$ are the knots at steps $k-1$
and $k$ in the LAR path, and $M^+_k$ is the random variable from Lemma  
\ref{lem:refined}.  The statistic in \eqref{eq:spacing} is one-sided,
implicitly aligned against the 
alternative $H_1 : v^T \theta > 0$, where $v$ is as in \eqref{eq:lar_v}.  
Since $v^T y = \lambda_k \geq 0$, the
denominator in \eqref{eq:lar_v} must have the same sign as  
\smash{$X_{j_k}^T P^\perp_{A_{k-1}} y$}, i.e., the same sign as 
\smash{$e_k^T X_{A_k}^+ y$}.  Hence 
\begin{equation*}
H_1 : \;\,v^T \theta > 0 \iff 
H_1 : \;\, \sign(e_k^T X_{A_k}^+ y) \cdot e_k^T X_{A_k}^+ \theta > 0,
\end{equation*}
i.e., the alternative hypothesis $H_1$ is that the 
population regression coefficient of the last selected variable is 
nonzero, and shares the sign of the sample regression coefficient of
the last variable.  This is a natural setup for a one-sided
alternative, as discussed in Section \ref{sec:oneor}.

The spacing test statistic falls directly out of our
polyhedral testing framework, adapted to the case of a random 
$U$ (Corollary \ref{cor:truncate} and Lemma \ref{lem:indep}).  It is a
valid p-value for testing $H_0 : v^T \theta = 0$, and has
exact conditional size. We emphasize this point by stating it in a
theorem.     

\begin{theorem}[\textbf{Spacing test}]
\label{thm:spacing}
Suppose that we have run $k$ steps of LAR. Represent the
conditions \eqref{eq:ochar}, \eqref{eq:pchar} in Lemma
\ref{lem:refined} as $\Gamma y \geq U$.  Specifically, we define 
$\Gamma$ to have the following $k+1$ rows:
\begin{align*}
&\Gamma_1 = c(j_1,s_1,A_0,s_{A_0})-c(j_2,s_2,A_1,s_{A_1}), \\
&\Gamma_2 = c(j_2,s_2,A_1,s_{A_1})-c(j_3,s_3,A_2,s_{A_2}), \\
&\ldots \\
&\Gamma_{k-1} =
c(j_{k-1},s_{k-1},A_{k-2},s_{A_{k-2}})-
c(j_k,s_k,A_{k-1},s_{A_{k-1}}), \\ 
&\Gamma_k = \Gamma_{k+1} = c(j_k,s_k,A_{k-1},s_{A_{k-1}}), 
\end{align*}
and $U$ to have the following $k+1$ components:
\begin{align*}
&U_1 = U_2 = \ldots = U_k = 0, \\
&U_{k+1} = M^+_k \Big(j_k,s_k, c(j_{k-1},s_{k-1},
  A_{k-2},s_{A_{k-2}})^T y\Big).
\end{align*}
For testing the null hypothesis $H_0 : e_k^T X_{A_k}^+ \theta =0$, the
spacing statistic $T_k$ defined in \eqref{eq:weight},
\eqref{eq:spacing} serves as an exact p-value conditional on
$\Gamma y \geq U$:
\begin{equation*}
\P_{e_k^T X_{A_k}^+ \theta =0} \Big(T_k \leq \alpha \,\Big|\, 
\Gamma y \geq U \Big) = \alpha,
\end{equation*}
for any $0 \leq \alpha \leq 1$.  
\end{theorem}

\begin{remark}
\label{rem:spacing}
The p-values from our polyhedral testing
theory depend on the truncation limits $\V^\mathrm{lo},
\V^\mathrm{up}$,
and in turn these depend on the polyhedral representation. For the
special polyhedron 
$\{y : \Gamma y \geq U\}$ considered in the theorem, it turns
out that $\V^\mathrm{lo}=M_k^+$ and $\V^\mathrm{up}=\lambda_{k-1}$,
which is fortuitous, as it means that no extra computation is  
needed to form  $\V^\mathrm{lo},\V^\mathrm{up}$ (beyond that already
needed for the path and \smash{$M_k^+$}). 
Furthermore, for the contrast vector $v$
in \eqref{eq:lar_v}, it turns out that $\|v\|_2=1/\weight_k$.
These two facts completely explain the spacing test statistic
\eqref{eq:spacing}, and the proof of Theorem \ref{thm:spacing},
presented in Appendix \ref{app:spacing}, reduces to checking these
facts.
\end{remark}

\begin{remark}
The event $\Gamma y \geq U$ is not exactly equivalent
to the LAR selection event at the $k$th step.  Recall that, as
defined, this only encapsulates the first part \eqref{eq:ochar},  
\eqref{eq:pchar} of a longer set of conditions
\eqref{eq:ochar}--\eqref{eq:schar} that provides the exact
characterization, as explained in Lemma \ref{lem:refined}.  However,
in practice, we have found that \eqref{eq:ochar}, \eqref{eq:pchar}
often provide a very reasonable approximation to the LAR selection
event. In most examples, the spacing p-values are either
close to those from the exact test for LAR, or exhibit even better
power.   
\end{remark}

\begin{remark}
A two-sided version of the spacing statistic in
\eqref{eq:spacing} is given by $T_k'=2 \cdot \min\{T_k,1-T_k\}$.  The
result in Theorem \ref{thm:spacing} holds for this two-sided
version, as well.  
\end{remark}

\subsection{Conservative spacing test}

The spacing statistic in \eqref{eq:spacing} is very simple and
concrete, but it still does depend on the random variable $M_k^+$.
The quantity $M_k^+$ is computable in $O(p)$ 
operations (see Appendix \ref{app:refined} for its definition), 
but it is not an output of standard software for computing the LAR
path (e.g., the R package {\tt lars}).
To further simplify matters, therefore, we might consider
replacing $M_k^+$ by the next knot in the LAR path, $\lambda_{k+1}$.
The motivation is that sometimes, but not always, $M_k^+$ and  
$\lambda_{k+1}$ will be equal. In fact, as argued in Appendix 
\ref{app:spacing_conservative}, it will always be true that 
$M_k^+ \leq \lambda_{k+1}$, leading us to a conservative version 
of the spacing test.

\begin{theorem}[\textbf{Conservative spacing test}]  
\label{thm:spacing_conservative}
After $k$ steps along the LAR path, define the modified spacing test
statistic   
\begin{equation}
\label{eq:spacing_mod}
\tilde{T}_k =
\frac{\Phi(\lambda_{k-1} \frac{\weight_k}{\sigma}) -
\Phi(\lambda_k \frac{\weight_k}{\sigma})}
{\Phi(\lambda_{k-1} \frac{\weight_k}{\sigma}) -
\Phi(\lambda_{k+1} \frac{\weight_k}{\sigma})}.
\end{equation}
Here, $\weight_k$ is as defined in \eqref{eq:weight}, and
$\lambda_{k-1},\lambda_k,\lambda_{k+1}$ are the LAR knots at steps
$k-1,k,k+1$ of the path, respectively. 
Let $\Gamma y \geq U$ denote the compact polyhedral representation of 
the spacing selection event step $k$ of the LAR path, as described in
Theorem \ref{thm:spacing}. Then \smash{$\tilde{T}_k$} is 
conservative, when viewed as a conditional p-value for testing the
null hypothesis \smash{$H_0 : e_k^T X_{A_k}^+ \theta =0$}: 
\begin{equation*}
\P_{e_k^T X_{A_k}^+ \theta =0} \Big(\tilde{T}_k \leq \alpha \,\Big|\,  
\Gamma y \geq U \Big) \leq \alpha,
\end{equation*}
for any $0 \leq \alpha \leq 1$.
\end{theorem}

\begin{remark}
It is not hard to verify that the modified statistic in
\eqref{eq:spacing_mod} is a monotone decreasing function of 
$\lambda_k-\lambda_{k+1}$, the spacing between LAR knots at steps
$k$ and $k+1$, hence the name ``spacing'' test.  Similarly, the exact 
spacing statistic in \eqref{eq:spacing} measures the magnitude of 
the spacing $\lambda_k-M_k^+$.
\end{remark}

\section{Empirical examples}
\label{sec:examples}
  
\subsection{Conditional size and power of FS and LAR tests} 

We examine the conditional type I error and power properties of the 
truncated Gaussian (TG) tests for FS and LAR, as well as the spacing
test for LAR, and the covariance test for LAR.  
We generated i.i.d.\ standard Gaussian predictors $X$ with $n=50$,
$p=100$, and then normalized each predictor (column of $X$) to have
unit norm. 
We fixed true regression coefficients $\beta^*= (5,-5, 0,\ldots
0)$, and we set $\sigma^2=1$.  For a total of 1000 repetitions, we
drew observations according to $y \sim N(X\beta^*,\sigma^2 I)$, 
ran FS and LAR, and computed p-values across the first 3 steps.
Figure \ref{fig:pvalues} displays the results in the form of QQ
plots.  The first plot in the figure shows the p-values at step 1,
conditional on the algorithm (FS or LAR) having made a correct
selection (i.e., having selected one of the first two variables). The
second plot shows the same, but at step 2.  The third plot shows
p-values at the step 3, conditional on the algorithm having made an  
incorrect selection.  

At step 1, all p-values display very good power, about 73\% at a 10\%
nominal type I error cutoff.  There is an interesting departure
between the tests at step 2: we see that the covariance and spacing
tests for LAR actually yield much better power than the exact TG
tests for FS and LAR: about 82\% for the former versus 35\% for the
latter, again at a nominal 10\% type I error level.  At step 3, the TG
and spacing tests produce uniform p-values, as desired; the covariance
test p-values are super-uniform, showing the conservativeness of this
method in the null regime.  

Why do the methods display such differences in power at step 2?  A
rough explanation is as follows. The spacing test, recall, is defined
by removing a subset of the polyhedral constraints for the 
conditioning event for LAR, thus its p-values are based on less   
conditioning than the exact TG p-values for LAR.  Because it
conditions on less, i.e., it uses a larger portion of the sample
space, it can 
deliver better power; and though it (as well as the covariance test)
is not theoretically guaranteed to control type I error in finite
samples, it certainly appears to do so empirically, seen in the third 
panel of Figure \ref{fig:pvalues}.
The covariance test is believed to behave more like the spacing test
than the exact TG test for LAR; this is based on an asymptotic 
equivalence between the covariance and spacing tests, given in 
Section \ref{sec:covtest}, and it explains their similarities in the
plots. 

\begin{figure}[hbtp]
\centering
\includegraphics[width=\textwidth]{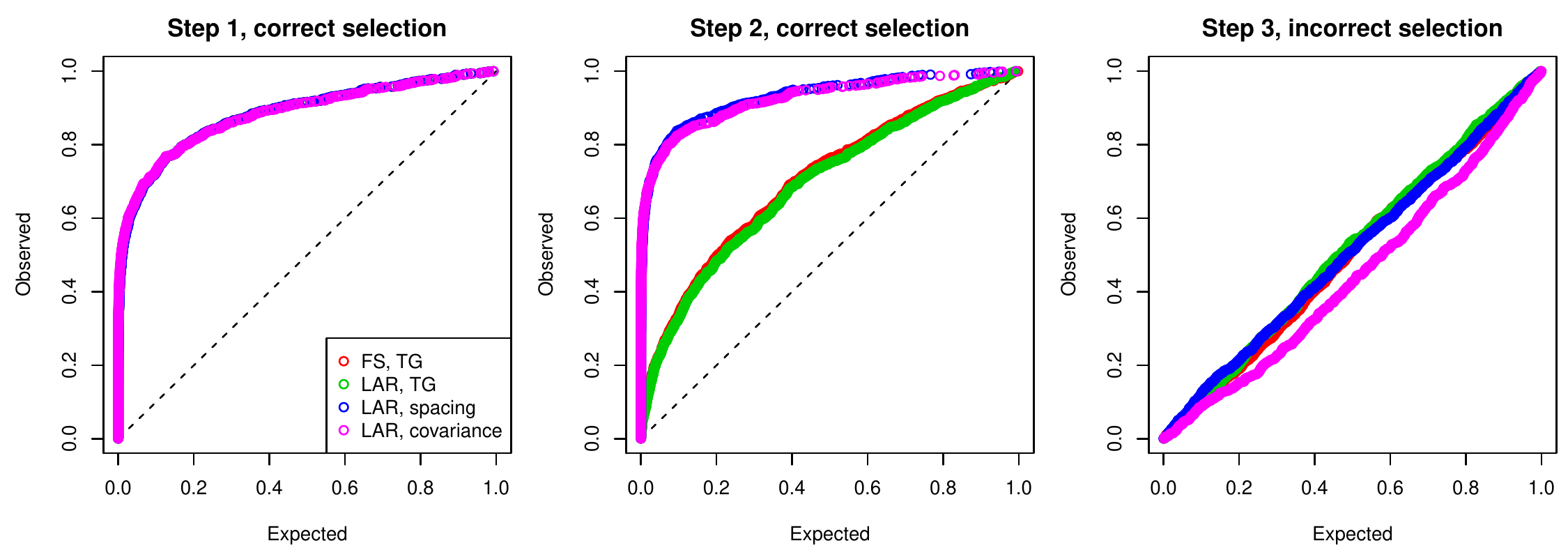}  
\caption{\it Simulated data with $n=50$, $p=100$, and two true active
  variables. Shown are p-values from the first 3 steps of FS and LAR,
  computed using the TG tests of Section \ref{sec:seq_tests}, the
  spacing test of Section \ref{sec:spacing}, and the covariance test
  of \citet{LTTT2013}, across 1000 repetitions (draws of $y$
  from the simulation model).}
\label{fig:pvalues}
\end{figure}

\subsection{Coverage of LAR conditional confidence intervals}  

In the same setup as the previous subsection, we computed
conditional confidence intervals over the first 3 steps of LAR,
at a 90\% coverage level.  Figure \ref{fig:intervals} shows these
intervals across the first 100 repetitions. Each interval here is
designed to cover the partial regression coefficient of a particular
variable, in a population regression of the mean $\theta=X\beta^*$ on
the variables in the current active set.  These population 
coefficients are drawn as black dots, and the colors of the intervals
reflect the identities of the variables being tested: red for variable  
1, green for variable 2, and blue for all other variables.  Circles
around population coefficients indicate that these particular
coefficients are not covered by their corresponding intervals.  The
miscoverage proportion is 12/100 in step 1, 11/100 in step 2, and
11/100 in step 3, all close to the nominal miscoverage level of 10\%.
An important remark: here we are counting marginal coverage of the 
intervals.  Our theory actually further guarantees {\it conditional}
coverage, for each model selection event; for example, among the red
intervals at step 1, the miscoverage proportion is 7/52, and among
green intervals, it is 5/47, both close to the nominal 10\% level.

\begin{figure}[htbp]
\centering
\includegraphics[width=\textwidth]{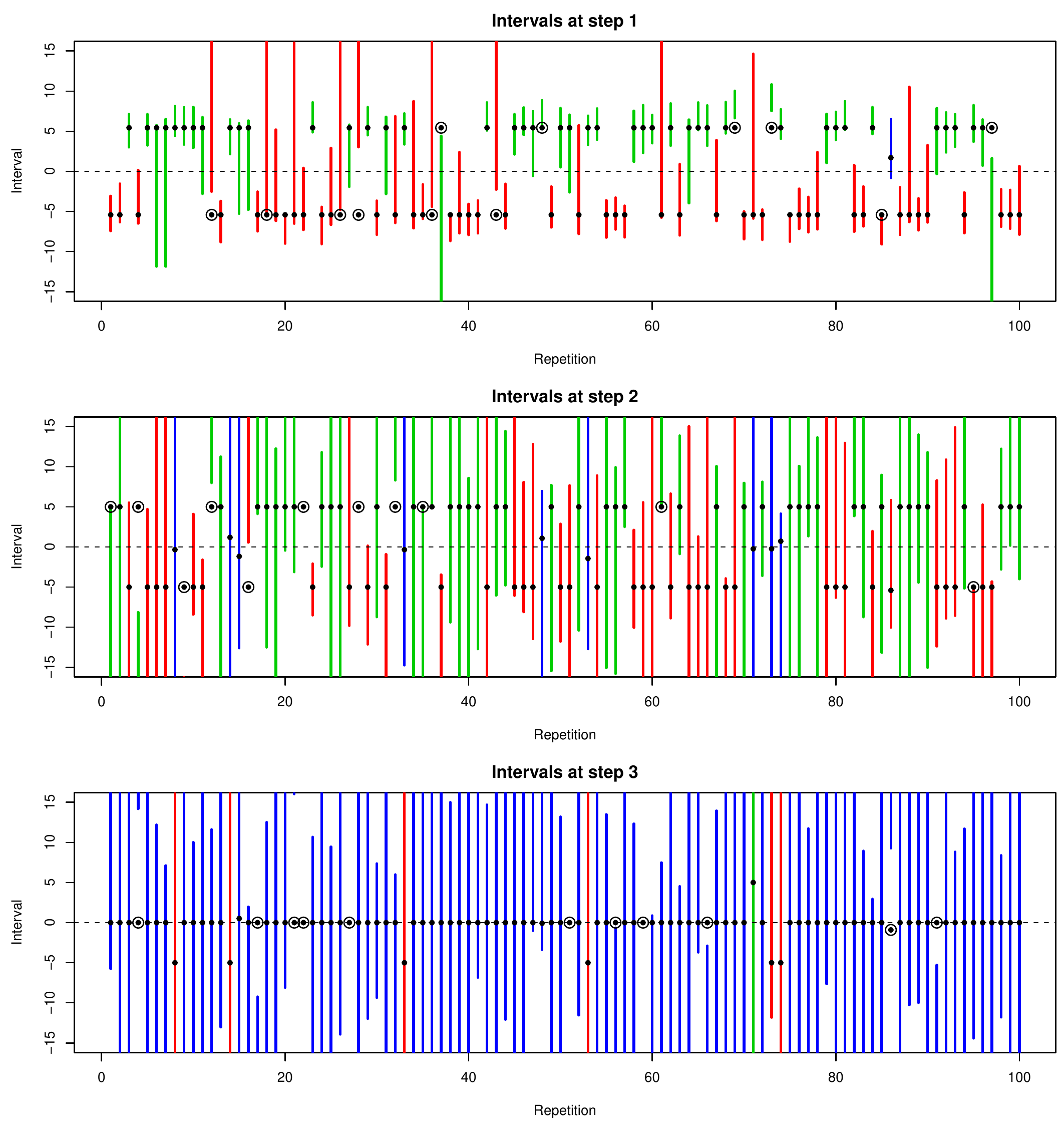}
\caption{\it Using the same setup as in Figure \ref{fig:pvalues},
  shown are 90\% confidence intervals for the selections made over the 
  first 3 steps of LAR, across 100 repetitions (draws of $y$ from the
  simulation model). The colors assigned to the intervals reflect the 
  identities of the variables whose partial regression coefficients
  are being tested: red for variable 1, green for variable 2, and blue
  for all others. The black dots are the true population coefficients,
  and circles around these dots denote miscoverages.
  The upper confidence limits for some of the parameters exceed the
  range for the y-axes on the plots (especially at step 3).}
\label{fig:intervals}
\end{figure}

\subsection{Comparison to the max-$|t|$-test}

The last two subsections demonstrated the unique properties of the
exact TG tests for FS and LAR.  For testing the significance of
variables entered by FS, \citet{buja2014} proposed what they call the 
{\it max-$|t|$-test}.  Here is a description. At the $k$th step of FS,
where $A_{k-1}$ is the current active list (with $k-1$ active
variables), let   
\begin{equation*}
t_{\max}(y)= \max_{j \notin A_{k-1}} \, 
\frac{| X_j^T P_{A_{k-1}}^\perp y |}
{\sigma \| P_{A_{k-1}} ^\perp X_j \|_2}. 
\end{equation*}
As the distribution of $t_{\max}(y)$ is generally intractable, 
we simulate $\epsilon \sim N(0,\sigma^2 I)$, and use this to estimate
null probability that $t_{\max}(\epsilon) > t_{\max}(y)$, which forms
our p-value.

We used the same setup as in the previous two subsections, but with an
entirely null signal, i.e., we set the mean to be
$\theta=X\beta^*=0$, in order to demonstrate the following point.  As
we step farther into the null regime (as we take more and more steps
with FS), the max-$|t|$-test becomes increasingly conservative,
whereas the exact TG test for FS continues to produce uniform
p-values, as expected.  The reason is that the TG test for FS at step
$k$ properly accounts for all selection events up to and including
step $k$, but the the max-$|t|$-test at step $k$ effectively ignores 
all selections occurring before this step, creating a conservative bias in 
the p-value. See Appendix \ref{app:maxt} for the plots.

\section{Relationship to the covariance test}
\label{sec:covtest}

There is an interesting connection between the LAR spacing test and
the covariance test of \citet{LTTT2013}. We first review the
covariance test and then discuss this connection. 

After $k$ steps of LAR, let $A_k$ denote the list of active variables
and \smash{$s_{A_k}$} denote the sign list, the same notation as we
have been using thus far.  The covariance test provides a significance
test for the $k$th step of LAR.  More precisely, it assumes an
underlying linear model $\theta=X\beta^*$, and tests the null
hypothesis 
\begin{equation*}
H_0 : \;\, A_{k-1} \supseteq \supp(\beta^*),
\end{equation*}
where $\supp(\beta^*)$ denotes the support of set of $\beta^*$ (the
true active set).  In words, this tests simultaneously the
significance of {\it any variable entered at step $k$ and later}.

Though its original definition is motivated from a difference in
the (empirical) covariance between LAR fitted values,
the covariance statistic can be written in an equivalent form
that is suggestive of a connection to the spacing test.  This form, at
step $k$ of the LAR path, is
\begin{equation}
\label{eq:covtest}
C_k = \weight_k^2 \cdot \lambda_k (\lambda_k - \lambda_{k+1}) 
/ \sigma^2, 
\end{equation}
where $\lambda_k,\lambda_{k+1}$ are the LAR knots at steps $k$ and
$k+1$ of the path, and $\weight_k$ is the weight in \eqref{eq:weight}.
(All proofs characterizing the null distribution of the covariance
statistic in \citet{LTTT2013} use this equivalent definition.)  The
main result (Theorem 3) in \citet{LTTT2013} is that, under
correlation restrictions on the predictors $X$ and other conditions,
the covariance statistic \eqref{eq:covtest} has a conservative
$\Exp(1)$ limiting distribution under the null hypothesis.  Roughly,
they show that
\begin{equation*}
\lim_{n,p \rightarrow \infty} \P_{A_{k-1} \supseteq \supp(\beta^*)}
\Big(C_k > t \, \Big| \, \hat{A}_k(y) =A_k, \, 
\hat{s}_{A_k}(y) =s_{A_k} \Big) \leq e^{-t},
\end{equation*}
for all $t \geq 0$.  

A surprising result, perhaps, is that the
covariance test in \eqref{eq:covtest} and the spacing test in
\eqref{eq:spacing_mod} are asymptotically equivalent.  
The proof for this equivalence uses relatively straightforward
calculations with Mills' inequalities, and is deferred until Appendix
\ref{app:spacing_equiv}. 

\begin{theorem}[\textbf{Asymptotic equivalence between spacing 
and covariance tests}]
\label{thm:spacing_equiv}
After a fixed number $k$ steps of LAR, the spacing
p-value in \eqref{eq:spacing_mod} and the covariance statistic
in \eqref{eq:covtest} are asymptotically equivalent, in the following
sense.  Assume an asymptotic regime in which 
\begin{gather*}
\weight_k \lambda_{k+1} \cp \infty, \;\;\;\text{and} \\
\weight_k^2 \cdot 
\lambda_{k-1}(\lambda_{k-1}-\lambda_k) \cp \infty, 
\end{gather*}
denoting convergence in probability. The spacing statistic,
transformed by the inverse $\Exp(1)$ survival function, satisfies
\begin{equation*}
-\log \Bigg(
\frac{\Phi(\lambda_{k-1} \frac{\weight_k}{\sigma}) -
\Phi(\lambda_k \frac{\weight_k}{\sigma})}
{\Phi(\lambda_{k-1} \frac{\weight_k}{\sigma}) -
\Phi(\lambda_{k+1} \frac{\weight_k}{\sigma})} \Bigg) =
\frac{\weight_k^2}{\sigma^2} \lambda_k (\lambda_k - \lambda_{k+1})
+ o_P(1).  
\end{equation*}
Said differently, the asymptotic p-value of the covariance statistic,
under the $\Exp(1)$ limit, satisfies
\begin{equation*}
\exp \bigg(-
\frac{\weight_k^2}{\sigma^2} \lambda_k (\lambda_k - \lambda_{k+1})
\bigg) = \Bigg(
\frac{\Phi(\lambda_{k-1} \frac{\weight_k}{\sigma}) -
\Phi(\lambda_k \frac{\weight_k}{\sigma})}
{\Phi(\lambda_{k-1} \frac{\weight_k}{\sigma}) -
\Phi(\lambda_{k+1} \frac{\weight_k}{\sigma})} \Bigg) ( 1+o_P(1)).
\end{equation*}
Above, we use $o_P(1)$ to denote terms converging to zero in 
probability. 
\end{theorem}

\begin{remark}
The asymptotic equivalence described in this theorem 
raises an interesting and unforeseen point about the one-sided nature
of the covariance test.  That is, the covariance statistic is seen to
be asymptotically tied to the spacing p-value in
\eqref{eq:spacing_mod}, which, recall, we can interpret as testing the
null $H_0 : e_k^T X_{A_k}^+ \theta =0$ against the one-sided
alternative \smash{$H_1 :\sign(e_k^T X_{A_k}^+ y)  
\cdot e_k^T X_{A_k}^+ \theta > 0$.}  The covariance test in
\eqref{eq:covtest} is hence implicitly aligned to have power when the 
selected variable at the $k$th step has a sign that matches that of   
the projected population effect of this variable.
\end{remark}

\section{Discussion}
\label{sec:discussion}

In a regression model with Gaussian errors, we have presented 
a method for exact inference, conditional on a polyhedral
constraint on the observations $y$.  Since the FS, LAR, and lasso
algorithms admit polyhedral representations for their 
model selection events, our framework produces exact p-values and  
confidence intervals post model selection for any of these 
adaptive regression procedures.  One particularly special and simple 
case arises when we use our framework to test the significance of the 
projected regression coefficient, in the population, of the latest
selected variable at a given step of LAR. This leads to the spacing
test, which is asymptotically equivalent to the covariance test of 
\citet{LTTT2013}.  An
R language package {\tt selectiveInference}, that implements the
proposals in this paper, is freely available on the CRAN repository,
as well as \url{https://github.com/selective-inference/R-software}.  A
Python implementation is also available, at 
\url{https://github.com/selective-inference/Python-software}. 

\subsection*{Acknowledgements} 

We would like to thank Andreas Buja, Max Grazier G'Sell,
Alessandro Rinaldo, and Larry Wasserman for helpful comments and
discussion.  We would also like to thank the editors and referees
whose comments led to a complete overhaul of this paper!
Richard Lockhart was supported by the Natural 
Sciences and Engineering Research Council of Canada; 
Jonathan Taylor was supported by NSF grant DMS 1208857 and AFOSR grant 
113039; Ryan Tibshirani was supported by NSF grant DMS-1309174;
Robert Tibshirani was supported by NSF grant DMS-9971405 and 
NIH grant N01-HV-28183.

\appendix

\section{Proofs}

\subsection{Proof of Lemma \ref{lem:refined}}
\label{app:refined}

We first consider an alternative characterization for the optimality
of a pair $(j_k,s_k)$ at an iteration $k$ of LAR. This
characterization is already proved in Lemma 7 of \citet{LTTT2013}, and
we present it here for reference.    

\begin{lemma}[Lemma 7 of \citealt{LTTT2013}]
\label{lem:ms}
Consider an iteration $\ell$ of LAR, and define
\begin{equation*}
\Sigma_{j,j'} = c(j,s,A_{\ell-1},s_{A_{\ell-1}})^T 
c(j',s',A_{\ell-1},s_{A_{\ell-1}}),
\end{equation*}
for variables $j,j' \notin A_{\ell-1}$ and signs $s,s'$.  (Note that
$\Sigma_{j,j'}$ actually depends on $s,s'$, but we suppress this
notationally for brevity.)  Also define
\begin{align*}
S_\ell^+(j,s,\rho) &= \Big\{ (j',s') : j' \notin A\cup\{j\}, \; 
1-\Sigma_{j,j'}/\Sigma_{jj} > 0,
\; c(j',s',A_{\ell-1},s_{A_{\ell-1}})^T y \leq \rho \Big\} \\
S_\ell^-(j,s,\rho) &= \Big\{ (j',s') : j' \notin A\cup\{j\}, \; 
1-\Sigma_{j,j'}/\Sigma_{jj} < 0,
\; c(j',s',A_{\ell-1},s_{A_{\ell-1}})^T y \leq \rho \Big\} \\
S_\ell^0(j,s,\rho) &= \Big\{ (j',s') : j' \notin A\cup\{j\}, \; 
1-\Sigma_{j,j'}/\Sigma_{jj} = 0,
\; c(j',s',A_{\ell-1},s_{A_{\ell-1}})^T y \leq \rho \Big\},
\end{align*}
and 
\begin{align*}
M_\ell^+(j,s,\rho) = \max_{(j',s') \in S_\ell^+(j,s,\rho)} \,
\frac{c(j',s',A_{\ell-1},s_{A_{\ell-1}})^T y -
  (\Sigma_{j,j'}/\Sigma_{j,j}) 
  \cdot  c(j,s,A_{\ell-1},s_{A_{\ell-1}})^T y}
{1 - (\Sigma_{j,j'}/\Sigma_{j,j})} \\
M_\ell^-(j,s,\rho) = \min_{(j',s') \in S_\ell^-(j,s,\rho)} \,
\frac{c(j',s',A_{\ell-1},s_{A_{\ell-1}})^T y -
  (\Sigma_{j,j'}/\Sigma_{j,j}) 
  \cdot  c(j,s,A_{\ell-1},s_{A_{\ell-1}})^T y}
{1 - (\Sigma_{j,j'}/\Sigma_{j,j})} \\
M_\ell^0(j,s,\rho) = \max_{(j',s') \in S_\ell^0(j,s,\rho)} \,
c(j',s',A_{\ell-1},s_{A_{\ell-1}})^T y - (\Sigma_{j,j'}/\Sigma_{j,j}) 
  \cdot  c(j,s,A_{\ell-1},s_{A_{\ell-1}})^T y.
\end{align*}
Then LAR selects $j_\ell$ and $s_\ell$ at iteration
$\ell$ if and only if
\begin{align}
\label{eq:charfirst}
c(j_\ell,s_\ell,A_{\ell-1},s_{A_{\ell-1}})^T y &\leq 
\lambda_{\ell-1}, \\ 
c(j_\ell,s_\ell,A_{\ell-1},s_{A_{\ell-1}})^T y &\geq 0, \\
c(j_\ell,s_\ell,A_{\ell-1},s_{A_{\ell-1}})^T y 
\label{eq:charmid}
&\geq M_\ell^+(j_\ell,s_\ell,\lambda_{\ell-1}), \\ 
c(j_\ell,s_\ell,A_{\ell-1},s_{A_{\ell-1}})^T y 
&\leq M_\ell^-(j_\ell,s_\ell,\lambda_{\ell-1}), \\
\label{eq:charlast}
0 & \geq M_\ell^0(j_\ell,s_\ell,\lambda_{\ell-1}).
\end{align}
Further, the triplet
\smash{$(M_\ell^+(j_\ell,s_\ell),M_\ell^-(j_\ell,s_\ell),
M_\ell^0(j_\ell,s_\ell))$} is  
independent of $c(j_\ell,s_\ell,A_{\ell-1},s_{A_{\ell-1}})^T y$, for
fixed $j_\ell,s_\ell$. 
\end{lemma}

We now recall another result of \citet{LTTT2013} that is important in
our context.   

\begin{lemma}[Lemma 9 of \citealt{LTTT2013}]
\label{lem:ord}
At any iteration $\ell$ of LAR, we have
\begin{equation*}
M_\ell^+(j_\ell,s_\ell,\lambda_{\ell-1}) \leq
c(j_{\ell+1},s_{\ell+1},A_\ell,s_{A_\ell})^T y. 
\end{equation*}
\end{lemma}

Finally, then, to build a list of constraints that are equivalent to
the LAR algorithm selecting variables $A_k$ and signs $s_{A_k}$
through step $k$, we intersect conditions
\eqref{eq:charfirst}--\eqref{eq:charlast} from Lemma \ref{lem:ms}
over $\ell=1,\ldots k$. Notice that for each $\ell<k$, the inequality 
in \eqref{eq:charmid} can be dropped, because it is implied by
\eqref{eq:charfirst} and the result of Lemma \ref{lem:ord},
\begin{equation*}
c(j_\ell,s_\ell,A_{\ell-1},s_{A_{\ell-1}})^T y \geq 
c(j_{\ell+1},s_{\ell+1},A_\ell,s_{A_\ell})^T y \geq 
M^+\Big(j_\ell,s_\ell,c(j_\ell,s_\ell,A_{\ell-1},s_{A_{\ell-1}})^T y\Big).
\end{equation*}
This gives rise to conditions \eqref{eq:ochar}--\eqref{eq:zchar} in
Lemma \ref{lem:refined}.  The last condition \eqref{eq:schar} is
needed to specify the sets $S_\ell$, $\ell=1,\ldots k$, and its
construction follows that given in Section \ref{sec:poly_lar}.

\subsection{Proof of Lemma \ref{lem:indep}}
\label{app:indep}

One can see from its definition in Lemma \ref{lem:ms} that $M_k^+$
is a maximum over linear functions of $y$ that are orthogonal to 
\smash{$\mathrm{col}(X_{A_k})$}, the span of active variables.   
Hence, when $v \in \mathrm{col}(A_k)$, we see that $M_k^+$ is a
function of $(I-vv^T/\|v\|_2^2) y$, certifying condition
\eqref{eq:ucond}. 

\subsection{Proof of Theorem \ref{thm:spacing}}
\label{app:spacing}


As pointed out in Remark \ref{rem:spacing}  after the theorem, we only 
need to prove 
that $\V^\mathrm{lo}=M_k^+$, $\V^\mathrm{up}=\lambda_{k-1}$, and
$\|v\|_2=1/\weight_k$, and the result will then follow from Corollary 
\ref{cor:truncate} and the polyhedral inference lemmas from Section
\ref{sec:poly}.  Well, $\V^\mathrm{lo}$ and $\V^\mathrm{up}$ are the
maximum and minimum of the quantities 
\begin{equation}
\label{eq:vquant}
\big(U_j - (\Gamma y)_j\big) \cdot \frac{\|v\|_2^2}{(\Gamma v)_j}  
+ \lambda_k,  
\end{equation}
over $j$ such that $(\Gamma v)_j > 0$ and $(\Gamma v)_j < 0$, 
respectively.  In the above, we have used the fact that $v^T y =
\lambda_k$, per its definition in \eqref{eq:lar_v}.

Now, since $v$ proportional to \smash{$P_{A_{k-1}}^\perp X_{j_k}$},
the first $k-1$ rows of $\Gamma$ are contained in
\smash{$\col(X_{A_{k-1}})$}, so   
$(\Gamma v)_j = 0$ for $j=1,\ldots k-2$.  For the $(k-1)$st row, we
have $(\Gamma v)_{k-1} = -\|v\|_2^2$, and for the $k$th and $(k+1)$st
row, we have $(\Gamma v)_k = (\Gamma v)_{k+1} = \|v\|_2^2$.  The
quantity \smash{$\V^\mathrm{up}$}, therefore, is singularly defined in 
terms of the $(k-1)$st row.  
As $U_{k-1}=0$ and $(\Gamma y)_{k-1} = \lambda_{k-1}-\lambda_k$, 
we can see from \eqref{eq:vquant} that
\smash{$\V^\mathrm{up}=\lambda_{k-1}$}.  The quantity
\smash{$\V^\mathrm{lo}$}, meanwhile, is defined in by the $k$th and
$(k+1)$st rows.  As $U_k=0$, $U_{k+1}=M_k^+$, and 
$(\Gamma y)_k = (\Gamma y)_{k+1} = \lambda_k$, we can see from
\eqref{eq:vquant} that \smash{$\V^\mathrm{lo}=M_k^+$}. 

Lastly, the result $\|v\|_2=1/\weight_k$ can be seen by direct
calculation, and is given in Lemma 10 of \citet{LTTT2013}.

\subsection{Proof of Theorem \ref{thm:spacing_conservative}} 
\label{app:spacing_conservative}


By Lemma \ref{lem:ord} of Appendix \ref{app:refined} (which recites  
Lemma 9 of \citet{LTTT2013}) we know that $\lambda_{k+1} \geq
M^+_k$. As the truncated Gaussian survival function is monotone  
increasing in its lower truncation limit, we have the relationship 
\smash{$\tilde{T}_k \geq T_k$}, for the modified and
original spacing statistics at step $k$.
The fact that the latter is uniform 
under the null implies that the former is super-uniform 
under the null, establishing the result.

\subsection{Plots for max-$|t|$-test comparison} 
\label{app:maxt}

Figure \ref{fig:maxt} displays the p-values from the TG test and the
max-$|t|$-test over 6 steps of FS, and 1000 repetitions (draws of $y$
from an entirely null model, with mean zero).  Both tests look good at
step 1, but the max-$|t|$-test becomes more and more conservative for
later steps. The reason is that TG test for FS conditions on all
selection events up to and including that at step $k$---the same is
true for the TG test for LAR, and its p-values would look exactly the
same here.  The max-$|t|$-test, however, does not do this. To be more
explicit, the max-$|t|$-test at step 2 ignores the fact  
that the observed $t_{\max}(y)$ is the {\it second largest} value of
the statistic in the data,  and erroneously compares it to a reference
distribution of {\it largest} $t_{\max}(\epsilon)$ values over a
smaller set.  This creates a conservative bias in the p-value.
Remarkably, the exact TG test for FS (and LAR) is able carry out full
conditioning on past selection events exactly. 

\begin{figure}[htb]
\centering
\includegraphics[width=\textwidth]{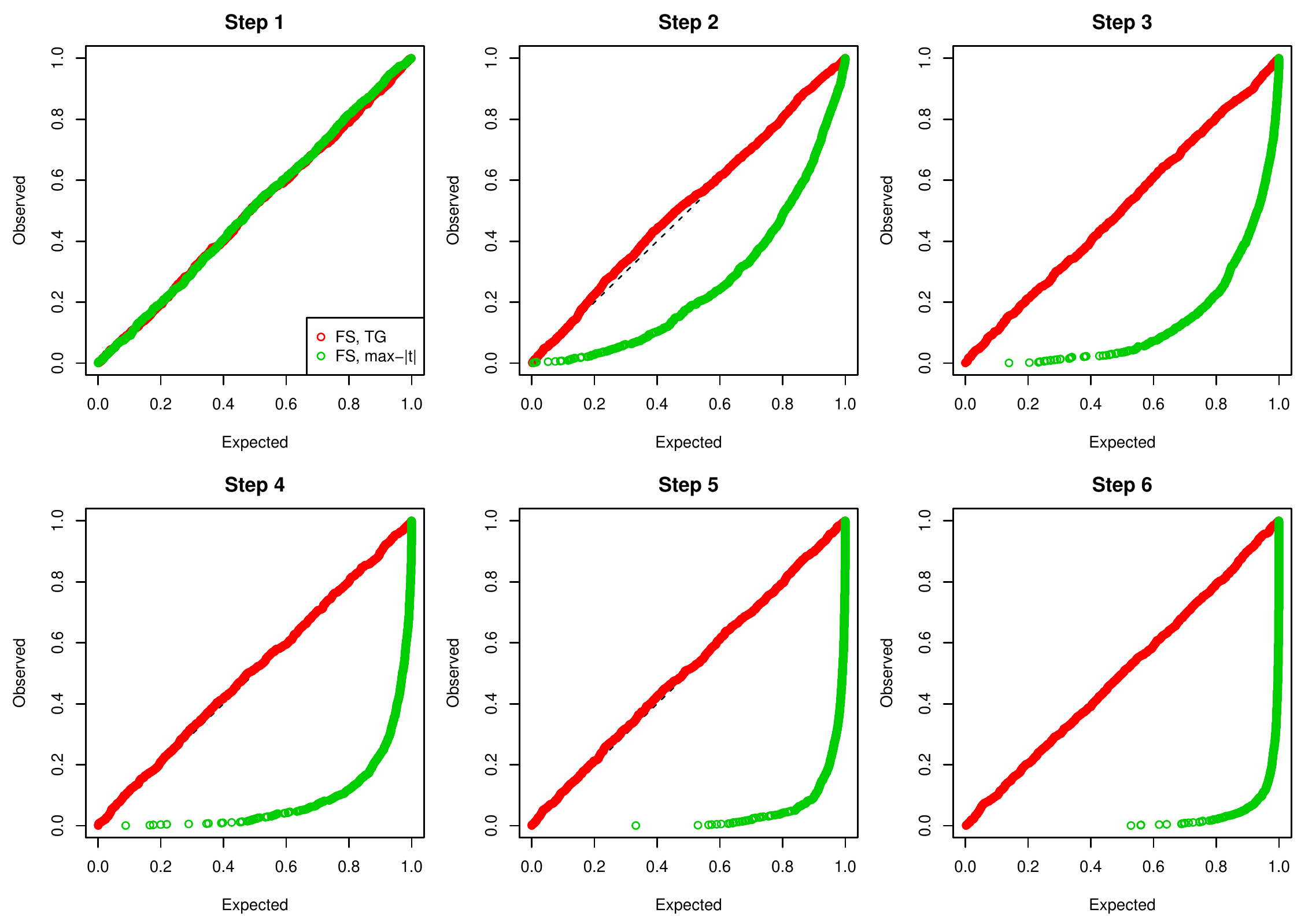} 
\caption{\it Using the same setup as in Figure \ref{fig:pvalues}, 
  but with the mean being $\theta=X\beta^*=0$, shown are p-values from 
  the first 6 steps of FS, using the TG test and the max-$|t|$-test,
  across 1000 repetitions (draws of $y$ from the simulation model).}   
\label{fig:maxt}
\end{figure}

\subsection{Proof of Theorem \ref{thm:spacing_equiv}}
\label{app:spacing_equiv}

We begin with a helpful lemma based on Mills' inequalities.

\begin{lemma}
\label{lem:helper}
Suppose that $x,y \rightarrow \infty$ and $y(y-x) \rightarrow   
\infty$. Then, for $\Phi$ the standard normal CDF and $\phi$ the
standard normal density,  
\begin{equation*}
\big(\Phi(y)-\Phi(x)\big) \cdot \frac{x}{\phi(x)} \rightarrow 1.
\end{equation*}
\end{lemma}
\begin{proof}
Write
\begin{equation*}
\Phi(y)-\Phi(x) = 1-\Phi(x) - \big(1-\Phi(y)\big),
\end{equation*}
and apply Mills' inequality to each of the survival terms
individually, yielding
\begin{equation*}
\frac{\phi(x)}{x} \frac{1}{1+1/x^2} - \frac{\phi(y)}{y} \leq
\Phi(y) - \Phi(x) \leq 
\frac{\phi(x)}{x} - \frac{\phi(y)}{y} \frac{1}{1+1/y^2}.
\end{equation*}
Multiplying through by $x/\phi(x)$, and noting that
\begin{equation*}
\frac{\phi(y)}{\phi(x)} = \exp(-(y+x)(y-x)/2)
\leq \exp(-y(y-x)/2) \rightarrow 0,
\end{equation*} 
we have established the result.
\end{proof}

Now consider
\begin{equation*}
\tilde{T}_k =
\frac{\Phi(\lambda_{k-1} \frac{\weight_k}{\sigma}) -
\Phi(\lambda_k \frac{\weight_k}{\sigma})}
{\Phi(\lambda_{k-1} \frac{\weight_k}{\sigma}) -
\Phi(\lambda_{k+1} \frac{\weight_k}{\sigma})}.
\end{equation*}
We assume 
\smash{$\weight_k \lambda_{k+1} \cp \infty$} (which implies 
the same of $\weight_k\lambda_k$, $\weight_k\lambda_{k-1}$), 
and \smash{$\weight_k^2 \lambda_{k-1}(\lambda_{k-1} - \lambda_k)     
\cp \infty$} (which implies the same of 
$\weight_k^2 \lambda_{k+1}(\lambda_{k-1} - \lambda_{k+1})$), hence
we can use Lemma \ref{lem:helper} to write 
\begin{equation*}
\tilde{T}_k = (1+o_P(1)) \cdot 
\frac{\phi(\lambda_k \frac{\weight_k}{\sigma})}
{\phi(\lambda_{k+1} \frac{\weight_k}{\sigma})} \cdot
\frac{\lambda_{k+1}}{\lambda_k},
\end{equation*}
where $o_P(1)$ denotes a term converging to zero in probability.  
Thus 
\begin{align*}
-\log(\tilde{T}_k) &= \frac{\weight_k^2}{\sigma^2} 
\frac{\lambda_k^2 - \lambda_{k+1}^2}{2} - 
\log \Big(\frac{\lambda_{k+1}}{\lambda_k}\Big) + o_P(1) \\ 
&= \frac{\weight_k^2}{\sigma^2}
\frac{\lambda_k(\lambda_k-\lambda_{k+1})}{2}  
+ \frac{\weight_k^2}{\sigma^2}
\frac{\lambda_{k+1}(\lambda_k-\lambda_{k+1})}{2}   
- \log \Big(\frac{\lambda_{k+1}}{\lambda_k}\Big) + o_P(1).
\end{align*}
Using the assumption that $\lambda_k/\lambda_{k+1} \rightarrow 1$ in
probability, this becomes
\begin{equation*}
-\log(\tilde{T}_k) = 
\frac{\weight_k^2}{\sigma^2} 
\lambda_k(\lambda_k-\lambda_{k+1}) 
+ o_P(1),
\end{equation*}
as desired.

\bibliographystyle{agsm}
\bibliography{ryantibs}

\end{document}